\documentclass[a4paper,11pt]{article}

\usepackage{verbatim}
\usepackage[T1]{fontenc}
\usepackage{lmodern}
\usepackage[all]{xy}
\usepackage{amsmath}
\usepackage{amsfonts}
\usepackage{amssymb}
\usepackage[latin1]{inputenc}
\usepackage{setspace,xypic}
\usepackage{indentfirst}
\usepackage{bbm}
\usepackage[all]{xy}

\usepackage{geometry}
\usepackage{hyperref}
\usepackage{times}
\usepackage{amsthm}
\usepackage{amsmath}
\usepackage{amsfonts}
\usepackage{amssymb}
\usepackage{bibentry}
\usepackage{color}
\usepackage{mathrsfs}
\usepackage{breqn}
\usepackage{stmaryrd}
\SetSymbolFont{stmry}{bold}{U}{stmry}{m}{n}
\newcommand {\emptycomment}[1]{}
\usepackage{bm}

\usepackage{tikz-cd}
\usepackage{tikz}
\usetikzlibrary{matrix,arrows,decorations.pathmorphing}
\usepackage{graphicx}
\usepackage{subfigure}
\usepackage[toc,page]{appendix}

\newcommand{\Dorfman}[1]{\{ #1\}}

\newcommand{\jet}{\mathfrak{J}}
\newcommand{\jetd}{\mathbbm{d}}
\newcommand{\dev}{\mathfrak{D}}
\newcommand{\dd}{\mathfrak{d}}
\newcommand{\ttt}{\mathfrak{t}}

\DeclareMathOperator{\Hom}{Hom}
\DeclareMathOperator{\End}{End}

\DeclareMathOperator{\gl}{gl}

\newcommand{\VB }{\mathsf{VB}}

\newtheorem{Thm}{Theorem}[section]

\newtheorem{Pro}[Thm]{Proposition}
\newtheorem{Lem}[Thm]{Lemma}
\newtheorem{Cor}[Thm]{Corollary}
\newtheorem{Def-Pro}[Thm]{Definition-Proposition}
\newtheorem{Def}[Thm]{Definition}

\theoremstyle{definition}
\newtheorem{Ex}[Thm]{Example}
\newtheorem{Rm}[Thm]{Remark}

\begin{document}
\title{Linearization of the higher analogue of Courant algebroids}

\vspace{-2cm}

\author{ \vspace{2mm}  Honglei Lang  and  Yunhe Sheng  }
\footnotetext{The research is supported by NSFC (11922110,11901568).}
\date{}
\maketitle
\makeatletter
\newif\if@borderstar
\def\bordermatrix{\@ifnextchar*{%
\@borderstartrue\@bordermatrix@i}{\@borderstarfalse\@bordermatrix@i*}%
}
\def\@bordermatrix@i*{\@ifnextchar[{\@bordermatrix@ii}{\@bordermatrix@ii[()]}}
\def\@bordermatrix@ii[#1]#2{%
\begingroup
\m@th\@tempdima8.75\p @\setbox\z@\vbox{%
\def\cr{\crcr\noalign{\kern 2\p@\global\let\cr\endline }}%
\ialign {$##$\hfil\kern 2\p@\kern\@tempdima & \thinspace %
\hfil $##$\hfil && \quad\hfil $##$\hfil\crcr\omit\strut %
\hfil\crcr\noalign{\kern -\baselineskip}#2\crcr\omit %
\strut\cr}}%
\setbox\tw@\vbox{\unvcopy\z@\global\setbox\@ne\lastbox}%
\setbox\tw@\hbox{\unhbox\@ne\unskip\global\setbox\@ne\lastbox}%
\setbox\tw@\hbox{%
$\kern\wd\@ne\kern -\@tempdima\left\@firstoftwo#1%
\if@borderstar\kern2pt\else\kern -\wd\@ne\fi%
\global\setbox\@ne\vbox{\box\@ne\if@borderstar\else\kern 2\p@\fi}%
\vcenter{\if@borderstar\else\kern -\ht\@ne\fi%
\unvbox\z@\kern-\if@borderstar2\fi\baselineskip}%
\if@borderstar\kern-2\@tempdima\kern2\p@\else\,\fi\right\@secondoftwo#1 $%
}\null \;\vbox{\kern\ht\@ne\box\tw@}%
\endgroup
}

\begin{abstract}
  In this paper,  we show that the spaces of sections of the $n$-th differential operator bundle $\dev^n E$ and the $n$-th skew-symmetric jet bundle $\jet_n E$ of a vector bundle $E$ are isomorphic to the spaces of linear $n$-vector fields and linear $n$-forms on $E^*$ respectively. Consequently,  the $n$-omni-Lie algebroid $\dev E\oplus\jet_n E$ introduced by Bi-Vitagliago-Zhang can be explained as certain linearization, which we call  pseudo-linearization    of   the higher analogue  of Courant algebroids $TE^*\oplus \wedge^nT^*E^*$. On the other hand,
we show that the omni $n$-Lie algebroid $\dev E\oplus \wedge^n\jet  E$ can also be explained as certain linearization, which we call   Weinstein-linearization  of   the higher analogue  of Courant algebroids $TE^*\oplus \wedge^nT^*E^*$.  We also show that $n$-Lie algebroids, local $n$-Lie algebras and Nambu-Jacobi structures can be characterized as integrable subbundles of omni $n$-Lie algebroids.
\end{abstract}

\tableofcontents

\vspace{-0.5cm}

\section{Introduction}

This paper aims to study  linearization of the higher analogue of Courant algebroids $TE^*\oplus \wedge^nT^*E^*$.

\subsection{Omni-Lie algebras and omni-Lie algebroids}

Courant algebroids were introduced by Liu, Weinstein and Xu in \cite{lwx} and have been found many applications in mathematical physics. See the survey article \cite{Kos} for more details. The notion of an omni-Lie algebra was introduced by Weinstein in \cite{Alan} to study the linearization of the standard Courant algebroid $TM\oplus T^*M$. Then it was further studied in \cite{kinyon-weinstein,ShengLiuZhu,UchinoOmni}. An {\bf omni-Lie algebra} associated to a vector space $V$ is a triple $(\gl(V)\oplus V,(\cdot,\cdot),\Dorfman{\cdot,\cdot})$, where $(\cdot,\cdot)$ is the $V$-valued pairing  given by
\begin{equation*}\label{eq:Vpair}
( A+u,B+v)=Av+Bu,\quad \forall ~A+u,B+v\in\gl(V)\oplus V,
\end{equation*}
and $\Dorfman{\cdot,\cdot}$ is the bilinear bracket operation given by
\begin{equation*}\label{omni}
\{A+u,B+v\}=[A,B]+Av.
\end{equation*}
 Note that $(\gl(V)\oplus V,\Dorfman{\cdot,\cdot})$ is not a Lie algebra, but a Leibniz algebra, which provides a natural example of Leibniz algebras. Moreover,  Dirac structures of the omni-Lie algebra $\gl(V)\oplus V$ characterize all Lie algebra structures on $V$, and this is one of the most important properties of an omni-Lie algebra. Let $M$ be the vector space $V^*$ in the standard Courant algebroid $TM\oplus T^*M$, and consider linear vector fields, which are in one-to-one correspondence with $\gl(V)$, and constant 1-forms on $V^*$, which are in one-to-one correspondence with $V$. Then the Dorfman bracket in the standard Courant algebroid $TM\oplus T^*M$ reduces to the bracket in the omni-Lie algebra $\gl(V)\oplus V$ given above. We use the terminology ``base-linearization'' to indicate this process.
Different generalizations of an omni-Lie algebra have been given recently with applications in different aspects.

 The notion of an omni-Lie algebroid was introduced in \cite{CLomni}, which can be viewed as the geometric generalization of an omni-Lie algebra from a vector space to a vector bundle. Lie algebroid structures on a vector bundle $E$ (or local Lie algebra structures when $E$ is a line bundle) can be characterized as Dirac structures of the omni-Lie algebroid $\dev E\oplus \jet E$, where $\dev E$ is the covariant differential operator bundle and $\jet E$ is the first jet bundle of $E$. Omni-Lie algebroids provide a general framework to study Jacobi structures, contact structures and odd dimensional  analogues  of generalized complex structures \cite{CLS2,IW,LSW,Luca18,Luca16,Wade}. Omni-Lie algebroids are also natural examples of $E$-Courant algebroids introduced in \cite{CLS}. Similar to the fact that $\gl(V)$ can be understood as linear vector fields on $V^*$, it is well-known that $\dev E$ corresponds to linear vector fields on the dual bundle $E^*$. But there are different explanations of $\jet E$:
 \begin{itemize}
   \item Sections of  $\jet E$ can be understood as constant $1$-forms on $E^*$. This explanation is supported by the fact that the pairing between $\dev E$ and $\jet E$ takes values in $E$, whose sections are linear functions on $E^*$. On the other hand, when $E$ reduces to a vector space $V$, we have $\jet V=V$, which is understood as the space of constant $1$-forms on $V^*$.  Therefore, this point of view is consistent with Weinstein's original idea in the study of linearization of the standard Courant algebroid $TM\oplus T^*M.$

       \item Sections of  $\jet E$ can also be understood as linear  $1$-forms on $E^*$. This explanation is supported by the fact that $\jet E$ corresponds to linear sections of the double vector bundle $(  T^*E^*;  E,E^*;M)$.
           \end{itemize}
When $\jet E$ is understood as constant $1$-forms on $E^*$, we will say that the omni-Lie algebroid $\dev E\oplus \jet E$ is the Weinstein-linearization of the standard Courant algebroid $TE^*\oplus T^*E^*$; when $\jet E$ is understood as linear  $1$-forms on $E^*$, we will say that the omni-Lie algebroid $\dev E\oplus \jet E$ is the pseudo-linearization of the standard Courant algebroid $TE^*\oplus T^*E^*$. Even though Weinstein-linearization and pseudo-linearization are the same in this situation, namely we both obtain the omni-Lie algebroid $\dev E\oplus \jet E$, in the sequel we will see that different geometric objects can be obtained using different explanations of $\jet E$.

  We summarize the above relations by the following diagram:
\[
		\begin{tikzcd}
			 TM\oplus T^*M \arrow{d}{\mbox{base-linearization}} \arrow{r}{} &~~~~~~TE^*\oplus T^*E^*\arrow{d}{\mbox{pseudo-(Weinstein-)linearization}}~~~~~~\\
			\mathrm{gl}(V)\oplus V  \arrow{r}{\mbox{geometrization}} &\dev E\oplus \jet E.
		\end{tikzcd}
\]

  \subsection{Omni $n$-Lie algebras and $n$-omni-Lie algebroids}
Recently,  the higher analogues of the standard Courant algebroid  $TM\oplus \wedge^{n}T^*M$ are widely studied due to  applications in Nambu-Poisson structures, multisymplectic structures, $L_\infty$-algebra theory and topological field theory  \cite{BS,BouwknegtJ,Grabowski,GS,hagiwara,hull,Zambon}. In \cite{LS}, the authors introduced the notion of an omni $n$-Lie algebra $\gl(V)\oplus \wedge^nV$ and proved that it is the base-linearization of the higher analogue of the standard Courant algebroid $TM\oplus \wedge^{n}T^*M$. Moreover, the $(n+1)$-Lie algebra structures on $V$ can be characterized as integrable subspaces of the omni $n$-Lie algebra $\gl(V)\oplus \wedge^nV$. $n$-Lie algebras (also called Filippov algebras) are the underlying algebraic structures of Nambu-Poisson structures,   and have many applications in mathematical physics. See the review article \cite{review} for more details. We summarize the above relation by the following diagram:
\[
		\begin{tikzcd}
			 TM\oplus T^*M  \arrow{d}{\mbox{base-linearization}} \ar{r}{} &TM\oplus \wedge^n T^*M\arrow{d}{\mbox{base-linearization}}\\
			\mathrm{gl}(V)\oplus V  \arrow{r}{} &\gl(V)\oplus \wedge^n V.
		\end{tikzcd}
\]

To study the higher analogue of the omni-Lie algebroid, the notion of an $n$-omni-Lie algebroid was introduced in \cite{BVZ}, which is the direct sum of the covariant differential operator bundle $\dev E$ and the $n$-th jet bundle $\jet_nE$ together with a pairing and a bracket operation. Multicontact structures can be seen as integrable subbundles of $n$-omni-Lie algebroids. Note that when the vector bundle $E$ reduces to a vector space $V$, one can not obtain the aforementioned omni $n$-Lie algebra since $\jet_nV=0$. On the other hand, $(n+1)$-Lie algebroid structures on $E$ can not be characterized by integrable subbundles of the $n$-omni-Lie algebroid  $\dev E\oplus \jet_n E$.

\subsection{Main results}

Since the omni-Lie algebroid $\dev E\oplus \jet E$ can be understood as certain linearization of the standard Courant algebroid $TE^*\oplus T^*E^*$, it is natural to expect that the $n$-omni-Lie algebroid given in \cite{BVZ} can be explained as certain linearization of the higher analogue of   Courant algebroids $TE^*\oplus \wedge^nT^*E^*$. On the other hand, due to limitations of the $n$-omni-Lie algebroid mentioned at the end of the last subsection, it is also natural to expect another geometric object that enjoys the following properties:
\begin{itemize}
  \item When the vector bundle reduces to a vector space, it reduces to the omni $n$-Lie algebra introduced in \cite{LS};
  \item $n$-Lie algebroid structures on a vector bundle can be characterized as certain integrable subbundles;
  \item It is  certain linearization of  $TE^*\oplus \wedge^nT^*E^*$.
\end{itemize}

The purpose of this paper is to answer the above questions. We   give an alternative explanation of the $n$-omni-Lie algebroid introduced in \cite{BVZ}. We find that
 linear $n$-vector fields and linear $n$-forms on a vector bundle $E$ are sections of the $n$-th differential operator bundle $\dev^n E$  and the $n$-th jet bundle $\mathfrak{J}_n E$ respectively. As a consequence, the $n$-omni-Lie algebroid $\dev E\oplus \jet_n E$ can be viewed as certain linearization, called   pseudo-linearization, of the higher analogue of Courant algebroid $T^*E\oplus \wedge^n T^*E^*$. On the other hand, if we understand $\jet E$ as constant $1$-forms on $E^*$, it is natural to use constant $n$-forms, that is $\wedge^n\jet E$, to replace $\jet_nE$  that used in \cite{BVZ}. Based on this observation,   the notion of an omni $n$-Lie algebroid is introduced that enjoys the above properties.  More precisely, an omni $n$-Lie algebroid associated to a vector bundle $E$ is the direct sum bundle $\dev E\oplus \wedge^n\jet E$ together with an anchor, a pairing and a bracket operation (see Definition \ref{defi:omninLie}). When $E$ reduces to a vector space $V$, we obtain the omni $n$-Lie algebra $\gl(V)\oplus \wedge^nV$ naturally. Moreover, we show that the omni $n$-Lie algebroid $\dev E\oplus \wedge^n\jet E$ can also be viewed as certain linearization, called Weinstein-linearization, of $TE^*\oplus \wedge^nT^*E^*$. When ${\rm rank} E\geq 2$ (resp. ${\rm rank} E=1$),  $(n+1)$-Lie algebroid structures (local $(n+1)$-Lie algebra structures) on $E$ can be characterized as  integrable subbundles of the omni $n$-Lie algebroid $\dev E\oplus \wedge^n\jet E$.    We summarize the relations by the following diagram:
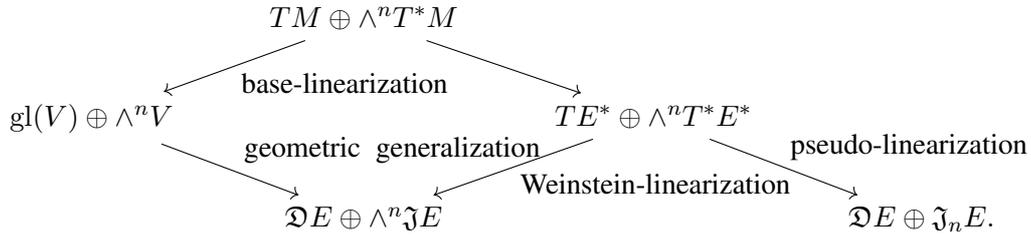
\[
		\begin{tikzcd}
			& TM\oplus \wedge^n T^*M \arrow{dl}{\mbox{base-linearization}}\arrow{dr}{} &&\\
\mathrm{gl}(V)\oplus\wedge^n V \arrow{dr}{\mbox{~geometric~ generalization~}}&&TE^*\oplus \wedge^n T^* E^*\arrow{dl}{\mbox{Weinstein-linearization}}\arrow{dr}{ \mbox{pseudo-linearization}}&\\
&\dev E\oplus \wedge^n \jet E&&\dev E\oplus \jet_n E.
		\end{tikzcd}
	\]

	We summarize   $n$-omni-Lie algebroids and omni $n$-Lie algebroids by the following table:

\begin{tabular}{|c|c|c|}
\hline
omni $n$-Lie algebroids&$n$-omni-Lie algebroids  \\ \hline
  $\dev E\oplus \wedge^n \jet E$ & $\dev E\oplus \jet_n E$\\ \hline
Weinstein-linearization of $TE^*\oplus \wedge^n T^* E^*$  & pseudo-linearization of $TE^*\oplus \wedge^n T^*E^*$ \\ \hline
$(n+1)$-Lie algebroid structures on $E$ when $\mathrm{rank} E\geq 2$ & higher Dirac-Jacobi structures\\ \hline
Nambu-Jacobi structures on $M$ & exact multi-symplectic structures\\ \hline
 Leibniz algebroid structures on $\wedge^n \jet E$ &- \\ \hline omni $n$-Lie algebra $\gl(V)\oplus \wedge^n V$ & - \\ \hline
\end{tabular}

\subsection*{Acknowledgments}
We would like to thank  Janusz  Grabowski, Zhangju Liu and Luca Vitagliago for helpful discussions and comments.

\section{Pseudo-linearization of $T E^*\oplus \wedge^n T^*E^*$ and $n$-omni-Lie algebroids}
The goal of this section is to give a geometric explanation of the $n$-omni-Lie algebroid $\dev E\oplus \jet_n E$ introduced in \cite{BVZ}. We understand it as a linearization of the higher analogue of Courant algebroids $T E^*\oplus \wedge^n T^*E^*$ in the sense that $\Gamma(\dev E)$ and $\Gamma(\jet_n E)$ are spaces of linear vector fields and linear $n$-forms (\cite{BC}) on the vector bundle $E^*$. Indeed, we give the geometric supports of the linear $n$-vector fields and linear $n$-forms on a vector bundle.

We first recall the $n$-th differential operator bundle $\dev^n E$ and the $n$-th skew-symmetric jet bundle $\jet_n E$ of a vector bundle $E$.

 A covariant differential operator for a vector bundle $E\to M$ is a smooth map $\dd: \Gamma(E)\to \Gamma(E)$, such that there is an element $X_\dd\in \mathfrak{X}^1(M)$, called the symbol, satisfying
\[ \dd(fu)=f\dd(u)+X_\dd(f)u,\qquad \forall~ f\in C^\infty(M),u\in \Gamma(E).\]
The covariant differential operator bundle  $\dev E$ of a vector bundle $E$
with the commutator bracket $[\cdot,\cdot]$ is a Lie algebroid, which is indeed the gauge Lie algebroid
of the frame bundle $\mathcal{F}(E)$.
The corresponding Atiyah sequence is as
follows:
\begin{equation}\label{Seg:DE}
		0\to \End(E)\xrightarrow{\mathbbm{i}} \dev E \xrightarrow{\mathbbm{j}} TM \to 0.
	\end{equation}
The first jet bundle $\jet E$ of a vector bundle $E$ is the bundle whose fiber over a point $m\in M$ is the space of equivalence classes of sections of $E$, where  $[u]_m=[v]_m$ for $u,v\in \Gamma(E)$  if $u(m)=v(m)$ and $d_m\langle u,\xi\rangle=d_m\langle v,\xi\rangle$ for any $\xi\in \Gamma(E^*)$.
In \cite{CLomni},
the authors proved that the first jet bundle $\jet E$  may be considered as
an $E$-dual bundle of $\dev E $, i.e.
\begin{equation*}\label{eqn:jetE}
\jet E \cong
\{\nu\in \Hom(\dev E,E)\,|\,
\nu(\Phi)=\Phi\circ\nu(\mathrm{Id}_E),\quad\forall~ \Phi\in \End(E)\}.
\end{equation*}

Associated to the jet bundle $\jet E$,  there is a jet sequence of $E$
given by:
\begin{equation}\label{Seq:JetE}
0\to \Hom(TM,E)\xrightarrow{\mathbbm{e}} \jet E \xrightarrow{\mathbbm{p}}E \to 0.
\end{equation}
This sequence does not necessarily split, but on the section level, it does:
\begin{eqnarray}\label{jetd}
\jetd: \Gamma(E) \rightarrow \Gamma(\jet E),\qquad \jetd u (\mathfrak{d}) := \mathfrak{d} (u),\qquad \forall ~u  \in   \Gamma(E), \mathfrak{d} \in\Gamma(\dev E).
\end{eqnarray}
A useful formula is
\begin{equation*}
\jetd(fu)=d f\otimes u+f\jetd u,\quad \forall~u\in\Gamma(E),~f\in C^\infty(M).
\end{equation*}
There is an $E$-valued pairing between $\jet E$ and $\dev E$ defined by
\begin{equation*}\label{Eqt:conpairingE}
\langle\mu,\mathfrak{d} \rangle:=\mathfrak{d}(u),\qquad  \mu\in (\jet{E})_m, \mathfrak{d}\in(\dev{E})_m,
\end{equation*}
where $u\in \Gamma(E)$ satisfies $\mu=[u]_m$. In particular, one has
\begin{eqnarray*}
\langle\mu,\Phi\rangle &=& \Phi\circ \mathbbm{p}(\mu),\quad\forall~\Phi\in \End(E),~\mu\in\jet{E};\\
\label{conpairing2} \langle \mathfrak{y},\mathfrak{d}\rangle&=&\mathfrak{y}\circ
\mathbbm{j}(\mathfrak{d}),\quad\forall~  \mathfrak{y}\in \Hom(TM,E),~\mathfrak{d}\in\dev{E}.
\end{eqnarray*}

The {\bf $n$-th differential operator bundle} $\dev^n E$ is introduced in \cite{CM, Sheng} as
\[\dev^n E:=\Hom(\wedge^n \jet E, E)_{\dev E}=\{\mathfrak{d}\in \Hom(\wedge^n \jet E, E)| \mathrm{Im}(\mathfrak{d}_\sharp) \subset \dev E\}, \qquad n\geq 2,\]
where $\mathfrak{d}_\sharp: \wedge^{n-1}  \jet E\to \Hom(\jet E,E)$ is defined by \[\mathfrak{d}_\sharp(\mu_1,\cdots,\mu_{n-1}) (\mu_n)=\mathfrak{d}(\mu_1,\cdots,\mu_n),\qquad \forall~ \mu_1,\cdots \mu_n\in \Gamma(\jet E).\]
When $\mathrm{rank} E\geq 2$, it fits into the following exact sequence:
\begin{eqnarray}\label{higher dev}
0\to \Hom(\wedge^n E,E)\xrightarrow{\mathbbm{f}} \dev^n E\xrightarrow{\mathbbm{q}} \Hom(\wedge^{n-1} E, TM)\to 0.
\end{eqnarray}

There is a graded Lie algebra structure on $\dev^\bullet E$ (\cite{CM}) given as follows:
\[[\dd_1,\dd_2]=(-1)^{(k+1)(l+1)}\dd_1\circ \dd_2-\dd_2\circ \dd_1\in \Gamma(\dev^{k+l-1} E),\qquad \dd_1\in \Gamma(\dev^k E),\dd_2\in \Gamma(\dev^l E),\]
where
\begin{eqnarray*}
\dd_2\circ \dd_1(\jetd u_1,\cdots,\jetd u_{k+l-1})=\sum_{\tau\in Sh(k,l-1)} (-1)^{\tau} \dd_2(\jetd (\dd_1(\jetd u_{\tau_1},\cdots,\jetd u_{\tau_k})),\jetd u_{\tau_{k+1}},\cdots,\jetd u_{\tau_{k+l-1}}),
\end{eqnarray*}
for $u_{i}\in \Gamma(E)$ and $\tau$ is taken over all $(k,l-1)$-shuffles.

In \cite{CLS}, the authors introduced the {\bf $n$-th skew-symmetric jet bundle}
\[\jet_n E:=\Hom(\wedge^n \dev E, E)_{\jet E}=\{\mu\in \Hom(\wedge^n \dev E, E)|\mathrm{Im}(\mu_\sharp)\subset \jet E\}, \qquad n\geq 2,\]
where $\mu_\sharp: \wedge^{n-1}\dev  E\to \Hom(\dev E, E)$ is the induced bundle map \[\mu_\sharp(\mathfrak{d}_1,\cdots \mathfrak{d}_{n-1})(\mathfrak{d}_n)=\mu(\mathfrak{d}_1,\cdots \mathfrak{d}_{n-1},\mathfrak{d}_n),\qquad \forall~ \mathfrak{d}_1,\cdots \mathfrak{d}_n\in \Gamma(\dev E).\]

Moreover, the $n$-th skew-symmetric jet bundle also fits into an exact sequence
\begin{eqnarray}\label{higher jet}
0\to \Hom(\wedge^n TM, E)\xrightarrow{\mathbbm{e}} \jet_n E\xrightarrow{\mathbbm{p}} \Hom(\wedge^{n-1} TM, E)\to 0.
\end{eqnarray}
There is a complex  $\mathbbm{d}:\jet_\bullet E\to \jet_{\bullet+1} E$. It is given as a subcomplex of the Chavelley-Eilenberg complex of the Lie algebroid $\dev E$ with the natural representation on the vector bundle $E$, whose differential $d_{\mathrm{CE}}:\Gamma(\Hom(\wedge^n \dev E,E))\to \Gamma(\Hom(\wedge^{n+1} \dev E,E))$ is defined  by
\begin{eqnarray}\label{cecomplex}
d_{\mathrm{CE}}\mu(\dd_1,\cdots,\dd_{n+1})&=&\nonumber \sum_{i=1}^{n+1}(-1)^{i+1}\dd_i\big(\mu(\dd_1,\cdots,\hat{\dd_i},\cdots,\dd_{n+1})\big)\\ &&+\sum_{i<j}(-1)^{i+j}\mu([\dd_i,\dd_j],\dd_1,\cdots,\hat{\dd_i},\cdots,\hat{\dd_j},\cdots,\dd_{n+1}),\quad \dd_i\in \Gamma(\dev E).
\end{eqnarray}
See \cite[Lemma 3.6]{CLS} for details.
\subsection{Linear $n$-vector fields on a vector bundle}

For a vector bundle $p_E:E\to M$,  identify $\Gamma(E^*)$ with the space of functions on $E$ which are linear along each fiber. A section of the first differential operator bundle $\dev E^*$ maps a section of $E^*$ to a section of $E^*$, which is viewed as a linear vector field on $E$. Denote by $\mathfrak{X}_{lin}^1(E)$ the space of linear vector fields on $E$. We have $\Gamma(\dev E^*)\cong
\mathfrak{X}^1_{lin}(E)$. We shall generalize this result to linear $n$-vector fields. Actually, linear $n$-vector fields  studied in \cite{BC,XuC} are isomorphic to sections of $n$-th differential operator bundle $\dev^n E$ introduced in \cite{CM, Sheng}.

An $n$-vector field $\pi\in \mathfrak{X}^n(E)$ is {\bf linear} if $\pi(d\phi_1,\cdots,d\phi_n)\in \Gamma(E^*)$ when $\phi_1,\cdots,\phi_n\in \Gamma(E^*)$ (\cite{XuC}). Note that, if $n> \mathrm{rank} E+1$, any $n$-vector field on $E$ is linear.

Assume that $ \mathrm{rank} E>1$ and let $(x^i,u^j)$ be a local coordinate of $E$, where $x^i$ and $u^j$ are the coordinate functions on $M$ and the fiber respectively. Then a linear $n$-vector field on $E$ has the local expression
\begin{eqnarray}\label{local n-vf}
\pi=\frac{1}{n!}\pi^{i_1\cdots i_n}_{j}(x)u^j \frac{\partial}{\partial u^{i_1}}\wedge \cdots\wedge  \frac{\partial}{\partial u^{i_n}}+\frac{1}{(n-1)!}\pi^{i_1\cdots i_{n-1},j}(x)\frac{\partial}{\partial u^{i_1}}\wedge\cdots \wedge\frac{\partial}{\partial u^{i_{n-1}}}\wedge \frac{\partial}{\partial x^j}.
\end{eqnarray}

Denote by $\mathfrak{X}_{lin}^n(E)$ the space of linear $n$-vector fields on $E$.  As explained in \cite{XuC}, a linear $n$-vector field $\pi$ for $n\geq 2$ has the  properties that
\begin{eqnarray}\label{pro for linear vf}
\pi(d\phi_1,\cdots,d\phi_{n-1}, dp_E^*f)=p_E^*g_{\phi_1,\cdots,\phi_{n-1},f}
\end{eqnarray} for some $g_{\phi_1,\cdots,\phi_{n-1},f}\in C^\infty(M)$ and
\[\iota_{dp_E^*{f_1}}\iota_{dp_E^*{f_2}} \pi=0,\qquad \forall ~\phi_1,\cdots \phi_{n-1}\in \Gamma(E^*), f, f_1,f_2\in C^\infty(M).\]
As a consequence, from a linear $n$-vector field $\pi$, we obtain $\delta_0:C^\infty(M)\to \Gamma(\wedge^{n-1} E)$ and $\delta_1:\Gamma(E)\to \Gamma(\wedge^n E)$ given by
\[\delta_0(f)(\phi_1,\cdots,\phi_{n-1}):=g_{\phi_1,\cdots,\phi_{n-1},f},\]
and \[\delta_1(u)(\phi_1,\cdots,\phi_n):=\sum_{i=1}^n(-1)^{i+n} g_{\phi_1,\cdots,\hat{\phi_i},\cdots,\phi_n,\phi_i(u)}-\langle\pi(d\phi_1,\cdots,d\phi_n),u\rangle,\]
for $u\in \Gamma(E)$.
This  correspondence is actually one-to-one; see \cite{XuC}.

We are now at the position to state our main result in this section. We find that the $n$-th differential operator bundle $\dev ^n E^*$ of $E^*$ is indeed the geometric support of linear $n$-vector fields on $E$:

\begin{Thm}\label{main1}
For a vector bundle $E$, the space of linear multivector fields $\mathfrak{X}^\bullet_{lin}(E)$ with the Schouten bracket $[\cdot,\cdot]_S$ is a graded Lie algebra. Moreover, we have the isomorphism \[(\Gamma(\dev^\bullet E^*), [\cdot,\cdot])\cong (\mathfrak{X}^\bullet_{lin}(E),[\cdot,\cdot]_S),\qquad \dd\mapsto \hat{\dd},\]
of graded Lie algebras, where $\hat{\dd}$ is determined by
\[\hat{\dd}(d\phi_1,\cdots,d\phi_n):=(-1)^n\dd(\jetd \phi_1,\cdots,\jetd \phi_n),\qquad \phi_i\in \Gamma(E^*),\]
where $\jetd: \Gamma(E)\to \Gamma(\jet E)$ is the natural map given by \eqref{jetd}.

\end{Thm}
\begin{proof}
First, we prove that linear multivector fields on $E$ are closed under the Schouten bracket, namely,
\[[\mathfrak{X}_{lin}^l(E),\mathfrak{X}_{lin}^k(E)]\subset \mathfrak{X}^{k+l-1}_{lin}(E).\]
By definition, for $X\in \mathfrak{X}_{lin}^l(E)$ and $Y\in \mathfrak{X}^k_{lin}(E)$, their Schouten bracket $[X,Y]_S$ is linear if \[[X,Y]_S(d\phi_1,\cdots,d\phi_{k+l-1})\in \Gamma(E^*)\] for any $\phi_i\in \Gamma(E^*)$. In fact,
we have
\begin{eqnarray*}
&&[X,Y]_S(d\phi_1,\cdots,d\phi_{k+l-1})\\ &=&\sum_{\sigma\in Sh(k,l-1)}(-1)^\sigma X(d(Y(d\phi_{\sigma_1},\cdots,d\phi_{\sigma_k})),d\phi_{\sigma_{k+1}},\cdots,d\phi_{\sigma_{k+l-1}})\\ &&-(-1)^{(k+1)(l+1)}\sum_{\tau\in Sh(l,k-1)}(-1)^\tau Y(d(X(d\phi_{\tau_1},\cdots,d\phi_{\tau_l})),d\phi_{\tau_{l+1}},\cdots,d\phi_{\tau_{k+l-1}}),
\end{eqnarray*}
where $\sigma$ and $\tau$ are taken over all the $(k,l-1)$-shuffles and $(l,k-1)$-shuffles respectively. As $Y$ is linear, we get that $Y(d\phi_{\sigma_1},\cdots,d\phi_{\sigma_k})\in \Gamma(E^*)$. Also since $X$ is linear, we further get that \[X(d(Y(d\phi_{\sigma_1},\cdots,d\phi_{\sigma_k})),d\phi_{\sigma_{k+1}},\cdots,d\phi_{\sigma_{k+l-1}})\in \Gamma(E^*)\] and thus $[X,Y]_S\in \mathfrak{X}^{k+l-1}_{lin}(E)$.
Thus we deduce that $(\mathfrak{X}^\bullet_{lin}(E),[\cdot,\cdot]_S)$ is a graded Lie subalgebra of $(\mathfrak{X}^\bullet(E),[\cdot,\cdot]_S)$.

Secondly, we check that $\dd\mapsto \hat{\dd}$ is an isomorphism of graded vector spaces. We shall find its inverse. For $X\in \mathfrak{X}_{lin}^n(E)$, define $\check{X}\in \Hom(\wedge^n \jet E^*, E^*)$ by
\[\check{X}(\jetd \phi_1,\cdots,\jetd \phi_n):=(-1)^nX(d\phi_1,\cdots,d\phi_n)\in \Gamma(E^*),\qquad \phi_i\in \Gamma(E^*).\]
The function linear property of $\check{X}$ requires that
\begin{eqnarray*}
 \check{X}(\jetd \phi_1,\cdots,\jetd \phi_{n-1}, df\otimes \phi_n)&=&(-1)^n
 X(d\phi_1,\cdots,d\phi_{n-1}, dp_{E^*}^*f)\phi_n,\\
 \check{X}(\jetd \phi_1,\cdots,\jetd \phi_{n-2}, dg\otimes \phi_{n-1},df\otimes \phi_n)&=&0.
 \end{eqnarray*}
 By  \eqref{pro for linear vf}, there exists a vector field $\mathbbm{j} (\check{X}_\sharp(d\phi_1,\cdots,d\phi_{n-1}))\in \mathfrak{X}^1(M)$ such that
  \begin{eqnarray*}
 \check{X}(\jetd \phi_1,\cdots,\jetd \phi_{n-1}, df\otimes \phi_n)
&=&\mathbbm{j} (\check{X}_\sharp(d\phi_1,\cdots,d\phi_{n-1}))(f)\phi_n\\ &=&
  (df\otimes \phi_n) \circ \mathbbm{j}(\check{X}_\sharp(\jetd \phi_1,\cdots,\jetd \phi_{n-1})).
 \end{eqnarray*}
  Thus we proved that $\check{X}\in \Gamma(\dev^n E^*)$. So we get a map $\mathfrak{X}^n_{lin}(E)\to \dev^n(E^*), X\to \check{X}$, which is actually  the inverse of the map $\dd\mapsto \hat{\dd}$.

At last, we show that
\[\widehat{[\dd,\ttt]}=[\hat{\dd},\hat{\ttt}]_S,\qquad \dd\in \Gamma(\dev^l E^*),\ttt\in \Gamma(\dev^k E^*).\]
Actually, we have
\begin{eqnarray*}
&&\widehat{[\dd,\ttt]}(d\phi_1,\cdots, d\phi_{k+l-1})\\ &=&(-1)^{k+l-1}
[\dd,\ttt](\jetd \phi_1,\cdots,\jetd \phi_{k+l-1})\\ &=&(-1)^{kl}\sum_{\sigma\in Sh(k,l-1)}(-1)^\sigma \dd(\jetd(\ttt(\jetd\phi_{\sigma_1},\cdots,\jetd\phi_{\sigma_k})),\jetd\phi_{\sigma_{k+1}},\cdots,\jetd\phi_{\sigma_{k+l-1}})\\ &&-(-1)^{k+l-1}\sum_{\tau\in Sh(l,k-1)}(-1)^\tau\ttt(\jetd(\dd(\jetd\phi_{\tau_1},\cdots,\jetd\phi_{\tau_l})),\jetd\phi_{\tau_{l+1}},\cdots,\jetd\phi_{\tau_{k+l-1}})\\ &=&[\hat{\dd},\hat{\ttt}]_S(d\phi_1,\cdots, d\phi_{k+l-1}).
\end{eqnarray*}
We thus proved that $\dd\mapsto \hat{\dd}$ defines an isomorphism of graded Lie algebras.
\end{proof}

By \cite{CM, Sheng}, the exact sequence \eqref{higher dev} always splits when $\mathrm{rank} E\geq 2$.
\begin{Cor}
If $\mathrm{rank} E\geq 2$, then we have
\[\mathfrak{X}_{lin}^n(E)\cong\Gamma(\dev^n E^*)\cong \Gamma(\wedge^n E\otimes E^*)\oplus \Gamma(\wedge^{n-1} E\otimes TM).\]

\end{Cor}
\begin{Ex}
When $E=TM$ for a manifold $M$, we have \[\mathfrak{X}_{lin}^n(TM)\cong \Gamma(\dev^n T^*M)\cong\mathfrak{X}^n(M)\otimes \Omega^1(M)\oplus (\mathfrak{X}^{n-1}(M)\otimes \mathfrak{X}^1(M)).\]
\end{Ex}
\begin{Ex}\label{cot}
When $E=T^*M$ for a manifold $M$, we have
\[\mathfrak{X}_{lin}^n(T^*M)\cong\Gamma(\dev^n TM)\cong \Omega^n(M)\otimes \mathfrak{X}^1(M)\oplus (\Omega^{n-1}(M)\otimes
\mathfrak{X}^1(M)).\]
\end{Ex}

\begin{Ex}
When $E=V^*$ is a vector space, we have $\dev V=\gl(V)$ and $\jet V=V$. In this case,
\[\mathfrak{X}_{lin}^n(V^*)\cong\Gamma(\dev^n V)=\Hom(\wedge^n V,V).\]
When $E=M\times V^*$, we have $\dev(M\times V)=TM\oplus \gl(V)$ and $\jet (M\times V)=(T^*M\otimes V) \oplus (M\times V)$. Furthermore,  the space of linear $n$-vector fields on $M\times V^*$ is
\[\mathfrak{X}^n_{lin}(M\times V^*)\cong\Gamma(\dev^n (M\times V))\cong
\Hom(\wedge^n V,V)\oplus (\mathfrak{X}^1(M)\otimes \wedge^{n-1} V^*).\]\end{Ex}

\begin{Ex}
Consider the case $E=M\times \mathbbm{R}$. Then we have $\dev E=TM\times \mathbbm{R}$ and $\jet E=T^*M \times\mathbbm{R}$. 
By definitions, we obtain
\[\mathfrak{X}_{lin}^n(M\times \mathbbm{R})\cong \Gamma(\dev^n (M\times \mathbbm{R}) )\cong \mathfrak{X}^n(M)\oplus \mathfrak{X}^{n-1}(M),\qquad n\geq 1. \]

\end{Ex}

\subsection{Linear $n$-forms on a vector bundle}
It is well known  that linear\footnote{Here ``linear'' is in the sense that $\jet E^*$ is the space of linear sections of the double vector bundle $(T^*E;E^*,E;M)$ over $E$ or in the sense of \cite{BC}.}   $1$-forms on a vector bundle $E$ can be viewed as sections of the first jet bundle $\jet E^*$. Here we find that linear $n$-forms studied in \cite{BC} are sections of the $n$-th jet bundle $\jet_n E$ introduced in \cite{CLS}.

\begin{Def}\label{linearkform}\rm{(\cite{BC})}
An $n$-form $\Lambda$ on a vector bundle $E$ is called {\bf linear} if the induced map
$\Lambda^\sharp: \oplus_E^{n-1} TE\to T^*E$:
\[
		\begin{tikzcd}
			\oplus_E^{n-1} TE \arrow{d}{} \ar{r}{\Lambda^\sharp} &T^*E\arrow{d}{}\\
			\oplus^{n-1} TM \ar{r}{\lambda} &E^*
		\end{tikzcd},
	\]
is a morphism of vector bundles, where $\lambda:\oplus^{n-1} TM\to E^*$ is the covering map on the base manifolds. The space of linear $n$-forms on $E$ is denoted by $\Omega_{lin}^n(E)$.
\end{Def}
As the map $\lambda$ is skew-symmetric, it is a bundle map $\wedge^{n-1} TM\to E^*$.
In particular, a linear $1$-form is a section of $T^*E\to E$ which induces a bundle map from $E\to M$ to $T^*E\to E^*$. In other words, it is a linear section of $T^* E\to E$ in the double vector bundle $(T^*E; E,E^*;M)$, which is a section of $\jet E^*$.

Choose a local coordinate $(x^i,u^j)$ on $E$, where $x^i$ and $u^j$ are the coordinate functions on $M$ and the fiber respectively. Then  an $n$-form on $E$ is linear  if and only if  locally it  has the formula (\cite{BC})
\begin{eqnarray}\label{linear n-form}
\Lambda=\frac{1}{n!}\Lambda_{i_1\cdots i_n,j}(x)u^j dx^{i_1}\wedge \cdots\wedge dx^{i_n}+\frac{1}{(n-1)!}\lambda_{i_1\cdots i_{n-1},j}(x)dx^{i_1}\wedge\cdots dx^{i_{n-1}}\wedge du^j.
\end{eqnarray}


We give another equivalent description of linear $n$-forms on $E$ by use of the linear vector fields on $E$.
\begin{Lem}\label{eq n-form}
An $n$-form $\Lambda\in \Omega^n(E)$ is linear if and only if there exists a bundle map $\lambda:\wedge^{n-1} TM\to E^*$, such that
\[\Lambda(X_1,\cdots,X_n)\in \Gamma(E^*),\qquad \Lambda(X_1,\cdots,X_{n-1},\Phi)=\Phi\circ \lambda(\underline{X_1},\cdots,\underline{X_{n-1}}),\]
where $X_i\in \mathfrak{X}^1_{lin}(E)$ which determines $\underline{X_i}\in \mathfrak{X}^1(M)$ and $\Phi\in \mathfrak{X}^1_{lin}(E)$ satisfying $\Phi(dp_E^*f )=0$ for any $f\in C^\infty(M)$.
\end{Lem}
\begin{proof}
Taking a local coordinate $(x^i,u^j)$ for $E$, a linear vector field $X\in\mathfrak{X}^1_{lin}(E)$ has the form
\[X=f_j^k(x) u^j \frac{\partial }{\partial u^k}+f^i(x)\frac{\partial}{\partial x^i}.\]
If $\Lambda$ is linear, by the local formula \eqref{linear n-form}, it is straightforward to see that $\Lambda(X_1,\cdots,X_n)\in \Gamma(E^*)$ for $X_i\in \mathfrak{X}^1_{lin}(E)$. Then suppose
\[X_l=f_{j_l}^{k_l}(x) u^{j_l} \frac{\partial }{\partial u^{k_l}}+f^{i_l}(x)\frac{\partial}{\partial x^{i_l}},\]
we have $\underline{X_l}=f^{i_l}(x)\frac{\partial}{\partial x^{i_l}}\in \mathfrak{X}^1(M)$. Assume
 $\Phi=f_j^k(x) u^j \frac{\partial}{\partial u^k}$, we have
\[\Lambda(X_1,\cdots,X_{n-1},\Phi)=\Lambda(f^{i_1}(x)\frac{\partial}{\partial x^{i_1}},\cdots,f^{i_{n-1}}(x)\frac{\partial}{\partial x^{i_{n-1}}},f_j^k(x) u^j \frac{\partial}{\partial u^k})=\Phi\circ \lambda(\underline{X_1},\cdots,\underline{X_{n-1}}).\]
It is similar for the converse.
\end{proof}

We shall show  that $\jet_n E^*$ serves as the geometric support of linear $n$-forms on a vector bundle $E$. The following theorem is a dual of Theorem \ref{main1}.
\begin{Thm}\label{main2}
For a vector bundle $E$, we have $d\Omega_{lin}^n (E)\subset \Omega_{lin}^{n+1} (E)$. Moreover, we have an isomorphism of cochain complexes:
\[(\Gamma(\jet_\bullet E^*),\jetd)\cong (\Omega_{lin}^\bullet (E),d),\qquad \mu\mapsto \hat{\mu},\]
where $\hat{\mu}$ is determined by \[\hat{\mu}(X_1,\cdots,X_n)=\mu(\check{X_1},\cdots,\check{X_n}),\qquad X_i\in \mathfrak{X}^1_{lin}(E)\] and $\check{X_i}\in \Gamma(\dev E^*)$ is decided by $\check{X_i}(\jetd \phi)=X_i(d\phi)$ for $\phi\in \Gamma(E^*)$.
\end{Thm}
\begin{proof}
Locally a linear $n$-form $\Lambda\in \Omega^n_{lin}(E)$ has the formula \eqref{linear n-form}. Taking the de Rham differential, we have
\begin{eqnarray*}
d\Lambda&=&\frac{1}{n!}\frac{\partial\Lambda_{i_1\cdots i_n,j}(x)}{\partial x^l} dx^l\wedge du^j\wedge dx^{i_1}\wedge \cdots\wedge dx^{i_n}+\frac{1}{n!}\Lambda_{i_1\cdots i_n,j}(x)du^j\wedge dx^{i_1}\wedge \cdots\wedge dx^{i_n}\\ &&+\frac{1}{(n-1)!}\frac{\partial \lambda_{i_1\cdots i_{n-1},j}(x)}{\partial x^k} dx^k \wedge dx^{i_1}\wedge\cdots dx^{i_{n-1}}\wedge du^j,
\end{eqnarray*}
which is still of the form \eqref{linear n-form}. So $d\Lambda$ is  a linear $(n+1)$-form on $E$. We get that $d\Omega_{lin}^n (E)\subset \Omega_{lin}^{n+1} (E)$.

Then we check that $\mu\mapsto \hat{\mu}$ is well-defined. It is obvious that $\hat{\mu}(X_1,\cdots,X_n)\in \Gamma(E^*)$.
And
\[\hat{\mu}(X_1,\cdots,X_{n-1},\Phi)=\mu(\check{X}_1,\cdots,\check{X}_{n-1},\Phi)= \Phi\circ \mathbbm{p}(\mu)(\underline{X}_1,\cdots,\underline{X_{n-1}}),\]
where $\mathbbm{p}:\jet_n E^*\to \Hom(\wedge^{n-1} TM,E^*)$ is the map in \eqref{higher jet}. So $\hat{\mu}\in \Omega^n_{lin}(E)$ and the associated bundle map $ \wedge^{n-1} TM\to E^*$ is $\mathbbm{p}(\mu)$.

Then we define an inverse map of $\mu\mapsto \hat{\mu}$. For $\Lambda\in \Omega^n_{lin}(E)$, define
\[\check{\Lambda}(\dd_1,\cdots,\dd_n)=\Lambda(\hat{\dd_1},\cdots,\hat{\dd_n}),\qquad \dd_i\in \Gamma(\dev E^*),\]
where $\hat{\dd_i}\in \mathfrak{X}^1_{lin}(E)$ is defined by $\hat{\dd_i}(d\phi)=\dd_i(\jetd \phi)$ for $\phi\in \Gamma(E^*)$.
By Lemma \ref{eq n-form}, we can check that $\check{\Lambda}\in \Gamma(\jet_n E^*)$. We get a map
\[\Omega^n_{lin}(E)\to \Gamma(\jet_n E^*),\qquad \Lambda\mapsto \check{\Lambda},\]
and it is the inverse of the map $\mu\mapsto \hat{\mu}$. We get an isomorphism.

Now it is left to check $\mu\mapsto \hat{\mu}$ actually gives an isomorphism of cochain complexes, namely
\[\widehat{\jetd \mu}=d\hat{\mu},\qquad \mu\in \Gamma(\jet_n E^*).\]
We have the following diagram:
\[
		\begin{tikzcd}
			(\Gamma(\jet_\bullet E^*),\jetd) \ar{d}{} \ar{r}{\subset} &(\Hom(\wedge^n \dev E^*,E^*),d_{\mathrm{CE}})
			\arrow{d}{}\\
			(\Omega^\bullet_{lin}(E),d) \ar{r}{\subset} &(\Hom(\wedge^n\mathfrak{X}^1_{lin}(E),\Gamma(E^*)),d)
		\end{tikzcd},
	\]
where the left two complexes are subcomplexes of the right two complexes.
We note that the right vertical side is actually an isomorphism  of cochain complexes.  As $\Gamma(\dev E)\cong \mathfrak{X}^1_{lin}(E)$, comparing the Chavelley-Eilenberg differential \eqref{cecomplex} and the de Rham differential,
we have $\widehat{\jetd \mu}=\widehat{d_{\mathrm{CE}}\mu}=d\hat \mu$ for $\mu\in \Gamma(\jet_n E^*)$.
\end{proof}
The exact sequence \eqref{higher jet} of $\jet_n E$ splits on the section level (\cite{CLS}).
\begin{Cor}
We have
\[\Omega_{lin}^n(E)\cong \Gamma(\jet_n E^*)\cong \Gamma(\wedge^n T^*M\otimes E^*)\oplus \Gamma(\wedge^{n-1} T^*M\otimes E^*).\]

\end{Cor}

\begin{Ex}
When $E=TM$ for a manifold $M$, we have \[\Omega_{lin}^n(TM)\cong\Gamma(\jet_n T^*M)\cong
\Omega^{n}(M)\otimes \Omega^1(M)\oplus (\Omega^{n-1} (M)\otimes \Omega^1(M)).\]
\end{Ex}
\begin{Ex}
When $E=T^*M$ for a manifold $M$, we have
\[\Omega_{lin}^n(T^*M)\cong \Gamma(\jet_n TM)\cong \Omega^n(M)\otimes \mathfrak{X}^1(M)\oplus (\Omega^{n-1}(M)\otimes \mathfrak{X}^1(M)).\]
Comparing with Example \ref{cot}, we see $\mathfrak{X}_{lin}^n(T^*M)\cong\Omega_{lin}^n(T^*M)$.\end{Ex}

\begin{Ex}
When $E=V^*$ is a vector space, we have
\[\Omega_{lin}^n(V^*)\cong \Gamma(\jet_n V)=0,\qquad n\geq 2.\]
Actually, for $\mu\in \Gamma(\jet_2 V)$, as $\mu(A\wedge B)=B\mu(A\wedge \mathrm{Id}_V)=-BA\mu(\mathrm{Id}_V\wedge \mathrm{Id}_V)=0$ for $A,B\in \mathrm{gl}(V)$, we see $\mu=0$.

When $E=M\times V^*$,  
we get
\[\Omega_{lin}^n(M\times V^*)\cong \Gamma(\jet_n (M\times V))\cong \Omega^n(M)\otimes V\oplus \Omega^{n-1}(M)\otimes V.\]\end{Ex}

\begin{Ex}
Consider the case $E=M\times \mathbbm{R}$, the trivial line bundle. We have $\dev E=TM\times \mathbbm{R}$ and $\jet E=T^*M \times \mathbbm{R}$. By definition, we obtain
\[\Omega_{lin}^n(M\times \mathbbm{R})\cong \Gamma(\jet_n E^*)\cong \Omega^{n}(M)\oplus \Omega^{n-1}(M).\]

\end{Ex}

\subsection{Pseudo-linearization of higher analogues of Courant algebroids $TE^*\oplus \wedge^n T^*E^*$}
 In this section, as consequences of Theorem \ref{main1} and \ref{main2}, we show that  the $n$-omni-Lie algebroid $\dev E\oplus \jet_n E$ introduced in \cite{BVZ} is  certain linearization of the higher analogue of Courant algebroids $TE^*\oplus \wedge^n T^*E^*$ (\cite{BS, Zambon}).

 For a manifold $M$, on the vector bundle
$\mathcal{T}^n:=TM\oplus \wedge^n T^*M,$
there exists a natural  non-degenerate symmetric pairing with values in $\wedge^{n-1} T^*M$:
\[(X+\alpha, Y+\beta)=\iota_X \beta+\iota_Y \alpha,\qquad \forall~X,Y\in \mathfrak{X}^1(M), \alpha,\beta\in \Omega^n(M),\]
an anchor map \[\rho:\mathcal{T}^n\to TM,\qquad \rho(X+\alpha)=X,\]
and a higher Dorfman bracket on $\Gamma(\mathcal{T}^n)$:
\begin{eqnarray}\label{HA}
\{X+\alpha, Y+\beta\}=[X,Y]+L_X \beta-\iota_Y d \alpha.
\end{eqnarray}
They satisfy some properties similar to that for a Courant algebroid. The quadruple $(\mathcal{T}^n,(\cdot,\cdot),\{\cdot,\cdot\},\rho)$ is called {\bf a higher analogue of Courant algebroids}  in \cite{BS,Zambon}.

The {\bf $n$-omni-Lie algebroid}  of a vector bundle $E$ (\cite{BVZ}) is the quadruple $(\dev E\oplus \jet_n E, (\cdot,\cdot),\{\cdot,\cdot\},\rho)$, where $\rho:\dev E \oplus \jet_n E\to \dev E$ is the projection to the first summand,
$(\cdot,\cdot)$ is the $\jet_{n-1} E$-valued pairing
\[(\mathfrak{d}+\mu,\mathfrak{t}+\nu)=\iota_{\mathfrak{d}} \nu +\iota_{\mathfrak{t}} \mu,\qquad \forall~\mathfrak{d},\mathfrak{t}\in \Gamma(\dev E), \mu,\nu\in \Gamma(\jet_n E),\]
and the bracket $\{\cdot,\cdot\}$ is
\[\{\mathfrak{d}+\mu,\mathfrak{t}+\nu\}=[\mathfrak{d},\mathfrak{t}]+L_{\mathfrak{d}} \nu-\iota_{\mathfrak{t}}\mathbbm{d} \mu.\]
Here $L_{\dd}: \Gamma(\jet_n E)\to \Gamma(\jet_n E)$ is defined in the following way: for $\nu\in \Gamma(\jet_n E)\subset \Hom(\wedge^n \dev E, E)$, suppose $\nu=\omega\otimes u$ for $\omega\in \Gamma(\wedge^n (\dev E)^*)$ and $u\in \Gamma(E)$. Define
\[L_\dd \nu=(L_\dd \omega)\otimes u+\omega\otimes \dd(u).\]
It is proved in \cite[Proposition 3.2]{CLS} that  $L_\dd \nu\in \Gamma(\jet_n E)$.

By Theorem \ref{main1} and \ref{main2}, linear vector fields and linear $n$-forms on $E^*$ can be seen as sections of $\dev E$ and $\jet_n E$ respectively. So linear sections of $TE^*\oplus \wedge^n T^* E^*$ are sections of the vector bundle $\dev E\oplus \jet_n E$.
Also, linear multivector fields and linear forms on a vector bundle are closed under the Schouten bracket and the de Rham differential respectively. The following lemma states that the linearity  is also preserved by the Lie derivative and the contraction.

\begin{Lem}
We have
\[\iota_{\mathfrak{X}^1_{lin}(E)} \Omega^n_{lin}(E)\subset \Omega_{lin}^{n-1}(E),\qquad L_{\mathfrak{X}^1_{lin}(E)} \Omega^n_{lin}(E)\subset \Omega_{lin}^n(E).\]
\end{Lem}
\begin{proof}
We prove by using the local coordinates. Choosing a local coordinate $(x^i,u^j)$ on $E$, suppose that $\Lambda\in \Omega^n_{lin}(E)$ takes the form of \eqref{linear n-form} and $X=f_j^k(x) u^j \frac{\partial }{\partial u^k}+f^i(x)\frac{\partial}{\partial x^i}\in \mathfrak{X}^1_{lin}(E)$. Then we have
\[\iota_X \Lambda=S(x) u^j dx^{i_1}\wedge \cdots\wedge dx^{i_{n-1}}+T(x)dx^{i_1}\wedge\cdots dx^{i_{n-2}}\wedge du^j,\]
for some $S(x),T(x)\in C^\infty(M)$. 
So $\iota_X \Lambda$ also takes the form of \eqref{linear n-form} and thus $\iota_X \Lambda\in
\Omega_{lin}^{n-1}(E)$.

Note that $L_X \Lambda=d\iota_X \Lambda+\iota_X d\Lambda$. By Theorem \ref{main2} and $\iota_{\mathfrak{X}^1_{lin}(E)} \Omega^n_{lin}(E)\subset \Omega_{lin}^{n-1}(E)$, i.e. the linearity is preserved by both the de Rham differential and the contraction, we obtain $L_X \Lambda\in \Omega^n_{lin}(E)$.
\end{proof}

Recall from Theorem \ref{main1} and \ref{main2} that we have the isomorphisms $\Gamma(\dev^\bullet E)\cong \mathfrak{X}_{lin}^\bullet(E^*),\dd\mapsto \hat{\dd}$ and $\Gamma(\jet_\bullet E)\cong \Omega^\bullet_{lin}(E^*),\mu\mapsto \hat{\mu}$.
\begin{Thm}\label{linear1}
For a vector bundle $E$, the $n$-omni-Lie algebroid $\dev E\oplus \jet_n E$ is induced from  the higher analogue of  Courant algebroids $(TE^*\oplus \wedge^n T^* E^*,(\cdot,\cdot),\{\cdot,\cdot\},\rho)$ by restricting to $\mathfrak{X}^1_{lin}(E^*)\oplus \Omega^n_{lin}(E^*)$. Precisely, we have
\begin{eqnarray*}
\widehat{(\dd,\mu)}&=&(\hat{\dd},\hat{\mu});\\
\widehat{[\dd,\ttt]}&=&[\hat{\dd},\hat{\ttt}]_S;\\
\widehat{L_\dd \mu}&=&L_{\hat{\dd}}\hat{\mu};\\
\widehat{\iota_\dd \jetd\mu}&=&\iota_{\hat{\dd}} d\hat{\mu},
\end{eqnarray*}
for $\dd,\ttt\in \Gamma(\dev E)$ and $\mu\in \Gamma(\jet_n E)$.
\end{Thm}
\begin{proof}
By Theorem \ref{main1} and \ref{main2}, we know that $\Gamma(\dev E)\cong\mathfrak{X}_{lin}^1(E^*)$, $\Gamma(\jet_n E)\cong\Omega^n_{lin}(E^*)$, and we have
$\widehat{[\dd,\ttt]}=[\hat{\dd},\hat{\ttt}]_S$.
We claim that
\begin{eqnarray}\label{dd1}
\iota_{\hat{\dd}}\hat{\mu}=\widehat{\iota_{\dd}\mu},\qquad \dd\in \Gamma(\dev E),\mu\in \Gamma(\jet_n E).
\end{eqnarray}
Indeed,
for $X_1,\cdots,X_{n-1}\in \mathfrak{X}_{lin}^1(E^*)$, we have
\[\iota_{\hat{\dd}} \hat{\mu}(X_1,\cdots,X_{n-1})=\hat{\mu}(\hat{\dd},X_1,\cdots,X_{n-1})=\mu(\dd,\check{X_1},\cdots,\check{X_{n-1}})=\widehat{\iota_{\dd} \mu}(X_1,\cdots,X_{n-1}),\]
where $\check{X_i}\in \Gamma(\dev E)$ is defined by $\check{X_i}(\jetd u)=X_i(du)$ for $u\in \Gamma(E)$.
We thus have $\widehat{(\dd,\mu)}=\widehat{\iota_\dd \mu}=\iota_{\hat{\dd}} \hat{\mu}=(\hat{\dd},\hat{\mu}).$
By Theorem \ref{main2}, we have
\begin{eqnarray}\label{dd2}
d\hat{\mu}=\widehat{\jetd \mu}.
\end{eqnarray}
By use of \eqref{dd1} and \eqref{dd2}, we have
\[\widehat{\iota_{\dd} \jetd \mu}=\iota_{\hat{\dd}} \widehat{\jetd\mu}=\iota_{\hat{\dd}}d \hat{\mu},\]
and
\[\widehat{L_\dd \mu}=\widehat{\iota_\dd \jetd \mu}+\widehat{\jetd \iota_\dd \mu}=\iota_{\hat{\dd}}\widehat{\jetd \mu}+d\widehat{\iota_\dd \mu}=\iota_{\hat{\dd}}d\hat{\mu}+d\iota_{\hat{\dd}} \hat{\mu}=L_{\hat{\dd}} \hat{\mu}.\]
This completes the proof.
\end{proof}
As a consequence, we call $n$-omni-Lie algebroids the {\bf pseudo-linearization} of higher analogues of Courant algebroids.
As the linearization, $n$-omni-Lie algebroids inherit many properties of the higher analogues of Courant algebroids.  By  \cite[Theorem 2.2, 2.5]{BS} and Theorem \ref{linear1}, we recover the following result in \cite{BVZ}, where they proved it by direct calculation.

A {\bf Leibniz algebroid}  (\cite{GM,ILM,Lodayalgebroid}) is a vector bundle $E$ with a bracket $\{\cdot,\cdot\}$ on $\Gamma(E)$ and a bundle map $\rho: E\to TM$, called the anchor map, satisfying that
\[\{u,\{v,w\}\}=\{\{u,v\},w\}+\{v,\{u,w\}\},\qquad \{u,fv\}=f\{u,v\}+\rho(u)(f) v,\]
for all $u,v,w\in \Gamma(E)$ and $f\in C^\infty(M)$.
\begin{Cor}
Let $(\dev E\oplus \jet_n E,(\cdot,\cdot),\{\cdot,\cdot\},\rho)$ be  an $n$-omni-Lie algebroid. Then
\begin{itemize}
\item[\rm{(i)}] $(\dev E\oplus \jet_n E,\{\cdot,\cdot\},\mathbbm{j}\circ \rho)$ is a Leibniz algebroid, where $\mathbbm{j}:\dev E\to TM$ is the map in \eqref{Seg:DE};
\item[\rm{(ii)}] $\{e,e\}=\frac{1}{2}\mathbbm{d}(e,e)$;
\item[\rm{(iii)}] $\rho(e_1)(e_2,e_3)=(\{e_1,e_2\},e_3)+(e_2,\{e_1,e_3\})$,
\end{itemize}
for all $e,e_i\in \Gamma(\dev E\oplus \jet_n E)$.
\end{Cor}

\begin{Rm}
When $E=V$, a vector space, the $n$-omni-Lie algebroid for $n\geq 2$ is just $\gl(V)$.
So $n$-omni-Lie algebroids do not include  omni $n$-Lie algebras studied in \cite{LS} as special cases. \end{Rm}

\section{Weinstein-linearization of $TE^*\oplus \wedge^n T^*E^*$ and omni $n$-Lie algebroids}

On the vector bundle $\dev E\oplus \wedge^n \jet E$, we introduce an $E\otimes \wedge^{n-1} \jet E$-valued pairing \begin{eqnarray}\label{pairing}
(\dd+\alpha, \ttt+\beta)=\iota_\dd \beta+\iota_\ttt \alpha,\qquad \forall ~\dd,\ttt\in \Gamma(\mathfrak{D} E),\alpha,\beta\in \Gamma(\wedge^n \mathfrak{J} E),\end{eqnarray}
and a bracket
\begin{eqnarray}\label{Dorf}
\{\dd+\alpha,\ttt+\beta\}:=[X,Y]+L_\dd \beta-\iota_\ttt \mathbbm{d}\alpha,
\end{eqnarray}
where $L_\dd :\Gamma(\wedge^n \jet E)\to \Gamma(\wedge^n \jet E)$ is defined by
\[L_\dd(\alpha_1\wedge\cdots\alpha_n)=\sum_{i=1}^n\alpha_1\wedge\cdots \wedge L_{\dd} \alpha_i\wedge\alpha_{i+1}\wedge\cdots\wedge \alpha_n,\qquad \alpha_i\in \Gamma(\jet E),\]

and $\jetd:\Gamma(\wedge^n \jet E)\to \Gamma(\jet_2 E\otimes \wedge^{n-1} \jet E)$ is defined by
\[\jetd(\alpha_1\wedge\cdots\wedge \alpha_n)=\sum_{i=1}^n (-1)^{i-1} (\jetd \alpha_i)\otimes \alpha_1\wedge\cdots \wedge \hat{\alpha_{i}}\wedge\cdots\wedge \alpha_n,\]
where $\hat{\alpha_i}$ means taking $\alpha_i$ out.
The contraction \[\iota_\ttt: \Gamma(\jet_2 E\otimes \wedge^{n-1} \jet E)\longrightarrow \Gamma(\wedge^n \jet E)\] is defined by
\begin{eqnarray}\label{key}
\iota_\ttt (\omega\otimes \alpha_1\wedge\cdots \wedge \alpha_{n-1}):= \omega(\ttt)\wedge \alpha_1\wedge\cdots \alpha_{n-1}+\sum_{i=1}^n (-1)^i\iota_\ttt  \alpha_i \otimes \omega\wedge\alpha_1\wedge\cdots \wedge \hat{\alpha_i}\wedge \cdots \wedge \alpha_{n-1},
\end{eqnarray} for $\omega\in \Gamma(\jet_2 E)$ and $\alpha_i\in \Gamma(\jet E)$.

The following lemma makes sure that $\iota_\ttt$ is well-defined.
\begin{Lem}
For $\ttt\in \Gamma(\dev E)$ and $\omega\otimes \alpha_1\wedge\cdots \wedge \alpha_{n-1} \in \Gamma(\jet_2 E\otimes \wedge^{n-1} \jet E)$, we have \[\iota_\ttt (\omega\otimes \alpha_1\wedge\cdots \wedge\alpha_{n-1})\in \Gamma(\wedge^n \jet E).\]
\end{Lem}

\begin{proof}
Let us prove this result by using Theorem \ref{main2}.
With respect to the relation between $\Omega^2_{lin}(E^*)$ and $\wedge^2 \Omega^1_{lin}(E^*)$, we have the following assertion:
\begin{eqnarray}\label{natural inclusion}
\Gamma(E)\otimes \Omega^2_{lin}(E^*) \hookrightarrow \wedge^2 \Omega^1_{lin}(E^*).
\end{eqnarray}
In fact,
it is direct from the formulas in local coordinates. Let $(x^i, u^j)$ be a local coordinate on $E^*$, where $x^i$ is the coordinate function on the base manifold $M$ and $u^j$ is the coordinate function on the fibers. Then  $S\in \Omega^1_{lin}(E^*)$ is of the form \[S=S_{i,j}(x) u^j dx^i+s_j(x)du^j,\] so a section $T\in \wedge^2 \Omega^1_{lin}(E^*)$ takes the form
\[T=\frac{1}{2}T_{j_1,j_2,i_1,i_2}(x) u^{j_1}u^{j_2} dx^{i_1}\wedge dx^{i_2}+\frac{1}{2}t_{j_1,j_2}(x)du^{j_1}\wedge du^{j_2}+t'_{j_1,j_2,i}(x) u^{j_1} dx^i\wedge du^{j_2}.
\]And a section $\Lambda\in \Omega^2_{lin}(E^*)$ has the form \[\Lambda=\frac{1}{2}\Lambda_{i_1,i_2,j}(x) u^j dx^{i_1}\wedge dx^{i_2}+\lambda_{i,j}(x)dx^i\wedge du^j.\]
Comparing the above two formulas, we get \eqref{natural inclusion}.

Thus, by Theorem \ref{main2}, we further get
\begin{eqnarray}\label{key1}
\Gamma(E\otimes \jet_2 E)\hookrightarrow \Gamma(\wedge^2 \jet E).
\end{eqnarray} This implies that $\iota_\ttt \alpha_i \otimes \omega$ in \eqref{key} belongs to $\Gamma(\wedge^2 \jet E)$.
So both of the two terms in the right hand side of \eqref{key}  belong to $\Gamma(\wedge^n \jet E)$. We obtain $\iota_\ttt (\omega\otimes \alpha_1\wedge\cdots \wedge\alpha_{n-1})\in \Gamma(\wedge^n \jet E)$.
\end{proof}
Based on this lemma, for $\ttt\in \Gamma(\dev E)$ and $\alpha\in \Gamma(\wedge^n \jet E)$, we have $\iota_\ttt \jetd \alpha\in \Gamma(\wedge^n \jet E)$. Since $L_\tau \alpha\in \Gamma(\wedge^n \jet E)$ by definition, we further get $\jetd \iota_\ttt \alpha\in \Gamma(\wedge^n \jet E)$.

\begin{Def}\label{defi:omninLie}
The {\bf omni $n$-Lie algebroid}\footnote{By personal communication with Luca Vitagliano, we learned that this definition was also given in their original version of \cite{BVZ}, but not appeared in the published version. } associated to a vector bundle $E$ is a  quadruple $(\mathfrak{D} E\oplus \wedge^n \mathfrak{J} E,(\cdot,\cdot),\{\cdot,\cdot\},\rho)$, where $\rho$ is the anchor map \[\rho:\dev E\oplus \wedge^n \jet E\to \dev E,\qquad \rho(X+\alpha)=X,\]
 the $E\otimes  \wedge^{n-1} \jet E$-valued pairing $(\cdot,\cdot)$ and the  bracket $\{\cdot,\cdot\}$ are given by \eqref{pairing} and \eqref{Dorf} respectively.
 \end{Def}

\subsection{Weinstein-linearization of higher analogues of Courant algebroids $TE^*\oplus \wedge^n T^*E^*$}

In this subsection, we show that the omni $n$-Lie algebroid given in the last subsection is certain linearization of the higher analogue of Courant algebroids $T E^*\oplus \wedge^n T^* E^*$.

 By Theorem \ref{main2}, we have the isomorphism $\Gamma(\wedge^n \jet E)\cong \wedge^n \Omega^1_{lin}(E^*), \alpha\mapsto \hat{\alpha}$, where
 \[\hat{\alpha}=\hat{\alpha_1}\wedge\cdots\wedge \hat{\alpha_n}\]
 if $\alpha=\alpha_1\wedge\cdots \wedge\alpha_n$ for $\alpha_i\in \Gamma(\jet E)$. Here $\hat{\alpha_i}\in \Omega^1_{lin}(E^*)$. \begin{Thm}\label{main3}
The omni $n$-Lie algebroid $\mathfrak{D} E\oplus \wedge^n \mathfrak{J} E$ associated to a vector bundle $E$ is the linearization of the higher analogue of  Courant algebroids $T E^*\oplus \wedge^n T^* E^*$ by restricting to $\mathfrak{X}_{lin}^1(E^*)\oplus \wedge^n \Omega_{lin}^1(E^*)$. Explicitly, we have
\begin{eqnarray*}
\widehat{(\dd,\alpha)}&=&(\hat{\dd},\hat{\alpha});\\
\widehat{[\dd,\ttt]}&=&[\hat{\dd},\hat{\ttt}]_S;\\
\widehat{L_\dd \alpha}&=&L_{\hat{\dd}}\hat{\alpha};\\
\widehat{\iota_\dd \jetd\alpha}&=&\iota_{\hat{\dd}} d\hat{\alpha},
\end{eqnarray*}
for $\dd,\ttt\in \Gamma(\dev E)$ and $\alpha\in \Gamma(\wedge^n \jet E)$.

\end{Thm}
\begin{proof}
Based on Theorem \ref{main1} and \ref{main2}, following the same manner as in Theorem \ref{linear1}, we could get this proof. We omit the details here.
\end{proof}

By Theorem \ref{main3} and the properties of higher analogues of Courant algebroids $TE^*\oplus \wedge^n T^*E^*$, we get the following relations.
\begin{Pro}
The omni $n$-Lie algebroid $(\dev E\oplus \wedge^n \jet E,(\cdot,\cdot),\{\cdot,\cdot\},\rho)$ has the properties:
\begin{itemize}
\item{ } $\{e_1,\{e_2,e_3\}\}=\{\{e_1,e_2\},e_3\}+\{e_2,\{e_1,e_3\}\}$;
\item{ } $\rho(\{e_1,e_2\})=[\rho(e_1),\rho(e_2)]$;
\item{ } $\{e_1,fe_2\}=f\{e_1,e_2\}+(\mathbbm{j}\circ \rho)(e_1)f e_2$;
\item{ } $\{e, e\}=\frac{1}{2} \mathbbm{d}(e,e)$;
\item{ } $\rho(e_1)(e_2,e_3)= (\{e_1,e_2\},e_3)+(e_2,\{e_1,e_3\})$,
\end{itemize}
for all $e_1,e_2,e_3\in \Gamma(\dev E\oplus \wedge^n \jet E)$,
where $\mathbbm{j}:\dev E\to TM$ is the map in \eqref{Seg:DE} and $\mathbbm{d}:\Gamma(E\otimes \wedge^{n-1} \jet E)\to \Gamma(\wedge^n \jet E)$ is defined by the map
\[\jetd(u\otimes \alpha_1\wedge \cdots \wedge \alpha_{n-1})=\jetd u \wedge \alpha_1\wedge \cdots \wedge \alpha_{n-1}+\sum_{i=1}^{n-1} (-1)^{i-1} u\otimes \mathbbm{d} \alpha_i \wedge \alpha_1\wedge \cdots \hat{\alpha_i}\wedge \cdots \wedge \alpha_{n-1},\] for $u\in \Gamma(E)$ and $\alpha_i\in \Gamma(\jet E)$, which is well-defined by \eqref{key1}.
\end{Pro}

\begin{Cor}
We have that $(\dev E\oplus \wedge^n \jet E,\{\cdot,\cdot\},\mathbbm{j}\circ \rho)$ is a Leibniz algebroid,
where $\mathbbm{j}:\dev E\to TM$ is the map in \eqref{Seg:DE}.
\end{Cor}


When the vector bundle $E$ is a vector space, denoted by $V$, we have $\dev E=\gl(V)$ and $\jet E=V$. In this case, the Lie derivative  of $\dev E$ on $\wedge^n \jet E$ and the contraction of $\wedge^n \jet E$ by $\dev E$ are:
\begin{eqnarray}\label{V}
L_X: \wedge^n V\to \wedge^n V,\qquad L_X(\alpha_1\wedge\cdots \wedge \alpha_n)=\sum_{i=1}^n \alpha_1\wedge \cdots \wedge X\alpha_i\wedge \cdots \alpha_n,
\end{eqnarray}
and
\begin{eqnarray}\label{V1}
\iota_X:\wedge^n V\to V\otimes \wedge^{n-1} V,\qquad \iota_X(\alpha_1\wedge\cdots \wedge \alpha_n)=\sum_{i=1}^n(-1)^{i-1}X\alpha_i\otimes \alpha_1\wedge\cdots\wedge \hat{\alpha_i}\wedge\cdots \wedge \alpha_n,
\end{eqnarray}
for $X\in \gl(V)$ and $\alpha_i\in V$.

The omni $n$-Lie algebroid for a vector space $V$ is \[(\gl(V)\oplus \wedge^n V,(\cdot,\cdot),\{\cdot,\cdot\}),\]
where the pairing $(\cdot,\cdot)$ takes values in $ V\otimes \wedge^{n-1} V$ and is given by
\[(X+\alpha,Y+\beta)=\iota_X \beta+\iota_Y \alpha,\qquad \forall~ X,Y\in \gl(V), \alpha,\beta\in \wedge^n V,\]
and the bracket $\{\cdot,\cdot\}$ is defined by
\[\{X+\alpha,Y+\beta\}=[X,Y]+L_X \beta.\]
Here $\iota_X \beta$ and $L_X \beta$ are defined by \eqref{V1} and \eqref{V}. This is the omni $n$-Lie algebra introduced in \cite{LS}, which is the base-linearization of the higher analogue of the standard Courant algebroid $TM\oplus \wedge^n T^*M$.

\subsection{Integrable subbundles of omni $n$-Lie algebroids and $n$-Lie algebroids}
The definition of $n$-Lie algebroids, also called Filippov algebroids,  was introduced in \cite{GM}. In this subsection, we show that the graph of $\Pi^\sharp: \wedge^n \jet E\to \dev E$ for $\Pi\in \Gamma(\dev^{n+1} E)$ is an integrable subbundle of the omni $n$-Lie algebroid $\dev E \oplus \wedge^n \jet E$ if and only if it defines an $(n+1)$-Lie algebroid structure on $E$.
\begin{Def}\rm{(\cite{GM})}
An $n$-Lie algebroid is a vector bundle $E$ with a skew-symmetric $n$-bracket on its sections \[[\cdot,\cdots,\cdot]:\Gamma(E)\times \cdots \times \Gamma(E)\to \Gamma(E)\]
satisfying the fundamental  identity:
\begin{eqnarray}\label{fi}
[u_1,\cdots,u_{n-1},[v_1,\cdots,v_n]]=\sum_{i=1}^n[v_1,\cdots,v_{i-1},[u_1,\cdots,u_{n-1},v_i],\cdots,v_n],\qquad \forall~u_i,v_i\in \Gamma(E),
\end{eqnarray}
and a bundle map $\rho: \wedge^{n-1} E\to TM$, called the anchor map, such that the Leibniz rule holds:
\[[u_1,\cdots,fu_n]=f[u_1,\cdots,u_n]+\rho(u_1,\cdots,u_{n-1})(f) u_n, \qquad \forall~u_i\in \Gamma(E),f\in C^\infty(M).\]
\end{Def}

When $E$ is a vector space, this recovers the notion of an {\bf $n$-Lie algebra} ( \cite{Filippov}).
The section space $\Gamma(E)$ of an $n$-Lie algebroid with the $n$-bracket $[\cdot,\cdots,\cdot]$ is an $n$-Lie algebra.

\begin{Rm}An $n$-Lie algebra structure on $V$ gives  rise to a Leibniz algebra structure on $\wedge^{n-1} V$ (\cite{DT}). Applying this to the section space $\Gamma(E)$ of an $n$-Lie algebroid $E$, we obtain a Leibniz algebra structure on
 $\Gamma(\wedge^{n-1} E)$ given by
\[\mathfrak{u}\circ \mathfrak{v}=\sum_{i=1}^{n-1}v_1\wedge\cdots v_{i-1}\wedge [u_1,\cdots,u_{n-1}, v_i]\wedge v_{i+1}\wedge\cdots \wedge v_{n-1},\]
for $\mathfrak{u}=u_1\wedge\cdots \wedge u_{n-1}$ and $\mathfrak{v}=v_1\wedge\cdots\wedge v_{n-1}$. Then we deduce that
the anchor map \[\rho: (\Gamma(\wedge^{n-1} E),\circ)\rightarrow (\mathfrak{X}^1(M), [\cdot,\cdot])\] in the definition of $n$-Lie algebroids is a Leibniz algebra morphism. This can be proved by replacing $v_n$ by $fv_n$ in the fundamental identity \eqref{fi} and then using the Leibniz rule.
This condition was listed in the definition of Filippov  algebroids in \cite{GM}, which is redundant.
\end{Rm}

Let $\Pi\in \Gamma(\dev^{n+1} E)$. By definition, it gives a bundle map $\Pi^\sharp:\wedge^{n} \jet E\to \dev E$. Conversely, a bundle map $\wedge^{n} \jet E\to \dev E$ gives a skew-symmetric map $\wedge^{n+1}\jet E\to E$, if and only if it  is a section of $\dev^{n+1} E$.

Define the graph of $\Pi^\sharp:\wedge^{n} \jet E\to \dev E$ by
\begin{eqnarray}\label{graph of pi}
\mathcal{G}_{\Pi^\sharp}=\{\Pi^\sharp(\alpha)+\alpha;\alpha\in \wedge^{n} \jet E\}\subset \dev E\oplus \wedge^{n} \jet E.
\end{eqnarray}

\begin{Thm}\label{k+1Leibniz}
If $E$ is a vector bundle with $\mathrm{rank} E\geq 2$, then there is a one-one correspondence between $(n+1)$-Lie algebroid structures $(E,[\cdot,\cdots,\cdot], \rho)$ and integrable subbundles $\mathcal{G}_{\Pi^\sharp}$ defined by \eqref{graph of pi} of the omni $n$-Lie algebroid $\dev E\oplus \wedge^n \jet E$ coming from  $\Pi\in \Gamma(\dev^{n+1} E)$.

Moreover, the bracket and anchor are related with $\Pi$  by
\begin{eqnarray}\label{leibniz}
[u_1,\cdots,u_{n+1}]&=&\Pi(\jetd u_1,\jetd u_2,\cdots,\jetd u_{n+1}),\\ \label{leibniz1} \rho(u_1,\cdots,u_{n})&=&\mathbbm{j}(\Pi^\sharp(\jetd u_1,\cdots,\jetd u_{n})),
\end{eqnarray}
where $u_i\in \Gamma(E)$ and $\mathbbm{j}:\dev E\to TM$ is the map in \eqref{Seg:DE}.
\end{Thm}

\begin{proof}
By Theorem \ref{main1}, for $\Pi\in \Gamma(\dev^{n+1} E)$, we have $\hat{\Pi}\in \mathfrak{X}^{n+1}_{lin}(E^*)$.  It is from \cite[Theorem 1]{BS} that $\hat{\Pi}$ is a Nambu-Poisson structure on $E^*$ if and only if the graph of $\hat{\Pi}^{\sharp}:\wedge^n T^* E^*\to TE^*$ is an integrable subbundle of the higher analogue of Courant algebroids $TE^*\oplus \wedge^n T^* E^*$.

By Theorem \ref{main3}, if the graph $\mathcal{G}_{\Pi^\sharp}\subset \dev E\oplus \wedge^n \jet E$ is integrable, then the graph $\mathcal{G}_{(-1)^{n+1}\hat{\Pi}^\sharp}\subset TE^*\oplus \wedge^n T^*E^*$ is also integrable, which is equivalent to that
$(-1)^{n+1}\hat{\Pi}$ is a Nambu-Poisson structure on $E^*$. In particular, this Nambu-Poisson structure is linear, meaning that
\[[u_1,\cdots, u_{n+1}]=(-1)^{n+1}\hat{\Pi}(du_1,\cdots,du_{n+1})=\Pi(\jetd u_1,\cdots,\jetd u_{n+1})\in \Gamma(E).\]
Therefore, we obtain an $(n+1)$-Lie algebra structure on $\Gamma(E)$.
Also, we have
\begin{eqnarray*}
[u_1,\cdots,u_n,fu_{n+1}]&=&\Pi(\jetd u_1,\cdots,f\jetd u_{n+1}+df\otimes u_{n+1})\\ &=&
f[u_1,\cdots,u_{n+1}]+\mathbbm{j}(\Pi^\sharp(\jetd u_1,\cdots,\jetd u_{n}))(f) u_{n+1}\\ &=&f[u_1,\cdots,u_{n+1}]+\rho(u_1,\cdots,u_n)(f) u_{n+1}.
\end{eqnarray*}
We claim that $(E,[\cdot,\cdots,\cdot],\rho)$ is an $(n+1)$-Lie algebroid.
It suffices to check that $\rho:\wedge^n \Gamma(E)\to \mathfrak{X}^1(M)$ defined by \eqref{leibniz1} induces a bundle map from $\wedge^n E$ to $TM$.  This is equivalent to
\begin{eqnarray}\label{anchor linear}
\mathbbm{j}(\Pi^\sharp(df\otimes u_1,\jetd u_2,\cdots, \jetd u_n))=0,\qquad \forall~f\in C^\infty(M), u_i\in \Gamma(E),\end{eqnarray}
as \[\rho(fu_1,\cdots,u_n)=\mathbbm{j}(\Pi^\sharp(df\otimes u_1,\jetd u_2,\cdots, \jetd u_n))+f\rho(u_1,\cdots,u_n).\]
In fact, as $\Pi$ is skew-symmetric, for $df\otimes u_1,df'\otimes u'_1\in \Gamma(T^*M\otimes E)$, we have
\begin{eqnarray*}
( \Pi^\sharp(df\otimes u_1,\jetd u_2,\cdots, \jetd u_n),df'\otimes u'_1)&=&\mathbbm{j} (\Pi^\sharp(df\otimes u_1,\jetd u_2,\cdots, \jetd u_n))(f') u'_1\\ &=&-\mathbbm{j} (\Pi^\sharp(df'\otimes u'_1,\jetd u_2,\cdots, \jetd u_n))(f)u_1.
\end{eqnarray*}
Since $\mathrm{rank} E \geq 2$, we choose $u_1$ and $u'_1$ to be independent. So the coefficients of $u_1$ and $u'_1$ in the above formulas must be zero. Thus
we proved \eqref{anchor linear}. Hence $(E,[\cdot,\cdots,\cdot],\rho)$ is an $(n+1)$-Lie algebroid.

From the above verification, the converse direction also holds.
\end{proof}

\begin{Pro}\label{jetLeibniz}
For $\Pi\in \Gamma(\dev^{n+1} E)$, its graph $\mathcal{G}_{\Pi^\sharp}$ given by \eqref{graph of pi} is an integrable subbundle of the omni $n$-Lie algebroid $\dev E\oplus \wedge^n \jet E$ if and only if
\[\Pi^\sharp[\alpha,\beta]_{\Pi}=[\Pi^\sharp(\alpha),\Pi^\sharp(\beta)],\qquad \forall~\alpha,\beta\in \Gamma(\wedge^n \jet E),\]
where the bracket $[\cdot,\cdot]_\Pi$ on  $\wedge^n \jet E$  is defined as
\[[\alpha,\beta]_\Pi=L_{\Pi^\sharp(\alpha)} \beta-L_{\Pi^\sharp(\beta)} \alpha+\jetd( \Pi^\sharp(\beta),\alpha).\]
Moreover, such an integrable subbundle induces a Leibniz algebroid $(\wedge^n \jet E,[\cdot,\cdot]_\Pi, \mathbbm{j}\circ \Pi^\sharp)$, where $\mathbbm{j}:\dev E\to TM$ is the bundle map in \eqref{Seg:DE}.
\end{Pro}
\begin{proof}
By the calculation
\[\{\Pi^\sharp(\alpha)+\alpha,\Pi^\sharp(\beta)+\beta\}=[\Pi^\sharp(\alpha),\Pi^\sharp(\beta)]+L_{\Pi^\sharp(\alpha)} \beta-\iota_{\Pi^\sharp(\beta)} \jetd \alpha,\qquad \forall~\alpha,\beta\in \Gamma(\wedge^n \jet E),\]
we see that $\mathcal{G}_{\Pi^\sharp}$ is closed with respect to the Dorfman bracket \eqref{Dorf} if and only if
 $\Pi^\sharp[\alpha,\beta]_{\Pi}=[\Pi^\sharp(\alpha),\Pi^\sharp(\beta)]$. This is equivalent to that $[\cdot,\cdot]_\Pi$ satisfies the Jacobi identity. It is left to  check the Leibniz rule. For $f\in C^\infty(M)$, we have
\[[\alpha,f\beta]_\Pi=f[\alpha,\beta]_\Pi+(L_{\Pi^\sharp (\alpha)} f) \beta=f[\alpha,\beta]_\Pi+\mathbbm{j}(\Pi^\sharp(\alpha))(f) \beta.\]
So $(\wedge^n \jet E,[\cdot,\cdot]_\Pi, \mathbbm{j}\circ \Pi^\sharp)$ is a Leibniz algebroid.
\end{proof}

For a Lie algebroid $E$, the first jet bundle $\jet E$ has a natural Lie algebroid structure with the bracket such that $[\jetd u_1,\jetd u_2]=\jetd [u_1,u_2]_E$.  Similar to this, we have the following result.

\begin{Pro}
Let  $(E,[\cdot,\cdots,\cdot]_E,\rho_E)$ be an $(n+1)$-Lie algebroid. Then
\begin{itemize}
\item[\rm (1)]  there exists a unique $(n+1)$-Lie algebroid structure $(\jet E,[\cdot,\cdots,\cdot],\rho_{\jet E})$ on $\jet E$ such that \[[\jetd u_1,\cdots,\jetd u_{n+1}]=\jetd [u_1,\cdots,u_n]_E,\qquad \rho_{\jet E}(\jetd u_1,\cdots,\jetd u_n)=\rho_E(u_1,\cdots,u_n);\]
\item[\rm (2)] there exists a unique Leibniz algebroid structure $(\wedge^n \jet E,[\cdot,\cdot],\rho)$ on $\wedge^n \jet E$ such that \[[\jetd u_1\wedge \cdots \wedge\jetd u_n, \jetd v_1\wedge \cdots \wedge\jetd v_n]=\sum_{i=1}^n \jetd v_1\wedge\cdots \jetd[u_1,\cdots,u_n, v_i]_E\wedge\cdots \wedge\jetd v_n,\]
and  $\rho(\jetd u_1\wedge \cdots \wedge\jetd u_n)=\rho_E(u_1,\cdots,u_n)$, for $u_i,v_i\in \Gamma(E)$.
\end{itemize}
\end{Pro}
\begin{proof}
The proof is standard. We omit the details.

\end{proof}

Applying Theorem  \ref{k+1Leibniz} and Proposition \ref{jetLeibniz} to the case when $E$ is a vector space $V$, we get a Leibniz algebra $\gl(V)\oplus \wedge^n V$.
Let $\Pi: \wedge^{n+1} V\to V$ be a skew-symmetric linear map. It induces a linear map $\Pi^\sharp: \wedge^n V\to \gl(V)$:\[\Pi^\sharp(\alpha)(\alpha_{n+1})=\Pi(\alpha,\alpha_{n+1}),\qquad \forall~ \alpha\in \wedge^n V,\alpha_{n+1}\in V.\]

\begin{Cor}
Let $\Pi: \wedge^{n+1} V\to V$ be a linear map. Then the following statements are equivalent:
\begin{itemize}
\item[\rm{(1)}]
 the graph $\mathcal{G}_{\Pi^\sharp}\subset \gl(V)\oplus \wedge^n V$ is a sub-Leibniz algebra;
 \item[\rm{(2)}] $V$  with the bracket  $\{\alpha_1,\cdots,\alpha_{n+1}\}=\Pi(\alpha_1,\cdots,\alpha_{n+1})$ is an $(n+1)$-Lie algebra;
\item[\rm{(3)}] $\wedge^n V$ with the bracket $[\alpha,\beta]_\Pi=L_{\Pi^\sharp(\alpha)} \beta$ is a Leibniz algebra;
\item[\rm{(4)}] $\Pi^\sharp[\alpha,\beta]_\Pi=[\Pi^\sharp(\alpha),\Pi^\sharp(\beta)]$, where $[\alpha,\beta]_\Pi:=L_{\Pi^\sharp(\alpha)} \beta$ for $\alpha,\beta\in \wedge^n V$.
\end{itemize}
\end{Cor}

The equivalence of $(1)$ and $(2)$ is exactly  \cite[Theorem 3.4]{LS}.

\subsection{Integrable subbundles of omni $n$-Lie algebroids and Nambu-Jacobi structures}

We explored the integrable subbundles of omni $n$-Lie algebroids when $\mathrm{rank}(E)\geq 2$ and found
Theorem \ref{k+1Leibniz}. Now we study the case when  $\mathrm{rank}(E)=1$ and particularly when $E$ is a trivial line bundle.

A {\bf local $n$-Lie algebra} is a vector bundle $E$ such that $\Gamma(E)$ has an $n$-Lie algebra structure with the property
$\mathrm{supp}[u_1,\cdots,u_n]\subset \mathrm{supp}u_1 \cap \cdots \cap \mathrm{supp} u_n$ for $u_i\in \Gamma(E)$. When $E$ is the trivial line bundle $M\times \mathbbm{R}$, it recovers the definition of Nambu-Jacobi structures on a manifold.

Nambu-Jacobi structures are the generalization of both Jacobi structures and Nambu-Poisson structures; see  \cite{hagiwara,H,MVV,MM}.
\begin{Def}
 A  Nambu-Jacobi structure of order $n$ on a manifold $M$ $(2\leq n\leq \mathrm{dim} M)$ is a linear skew-symmetric $n$-bracket $[\cdot,\cdots,\cdot]:C^\infty(M)\times \cdots \times C^\infty(M)\to C^\infty(M)$, which is a first differential operator, i.e.
\[[g_1g_2,f_1,\cdots, f_{n-1}]=g_1[g_2,f_1,\cdots, f_{n-1}]+g_2 [g_1,f_1,\cdots,f_{n-1}]-g_1g_2[1,f_1,\cdots, f_{n-1}],\]
and satisfies the fundamental identity \eqref{fi}, i.e.
\[[f_1,\cdots,f_{n-1},[g_1,\cdots,g_n]]=\sum_{i=1}^n[g_1,\cdots,g_{i-1},[f_1,\cdots,f_{n-1},g_i],g_{i+1},\cdots,g_n],\]
for $f_i,g_i\in C^\infty(M)$.
\end{Def}
A manifold with such a bracket on its function space  is called a {\bf Nambu-Jacobi manifold of order $n$}.  A {\bf Jacobi manifold} is a Nambu-Jacobi manifold of order $2$ and  a {\bf Nambu-Poisson manifold} is a Nambu-Jacobi manifold whose bracket vanishes if one of the functions is constant.  When $n>2$,
a Nambu-Jacobi structure of order $n$  is equivalently determined by a compatible pair of Nambu-Poisson structures $\Lambda\in \mathfrak{X}^n(M)$ and $\Gamma\in \mathfrak{X}^{n-1}(M)$; see \cite{IL,MVV} for details.

\begin{Thm}\label{rank1case}
Assume  that $E$ is a line bundle. There is a one-one correspondence between local $(n+1)$-Lie algebra structures $[\cdot,\cdots,\cdot]$ on $E$ and integrable subbundles $\mathcal{G}_{\Pi^\sharp}$ defined by \eqref{graph of pi} of the omni $n$-Lie algebroid $\dev E\oplus \wedge^n \jet E$ coming from $\Pi\in \Gamma(\dev^{n+1} E)$ satisfying that
\[[u_1,\cdots,u_{n+1}]=\Pi(\jetd u_1,\cdots,\jetd u_{n+1}),\qquad \forall ~ u_i\in \Gamma(E).\]

Moreover, $\Pi$ induces an $(n+1)$-Lie algebroid structure on $E$ if and only if $\mathbbm{j} \circ \Pi^\sharp \circ \mathbbm{e}=0$, where $\mathbbm{j}:\dev E\to TM$ and $\mathbbm{e}: T^*M\otimes E\to \jet E$ are given in \eqref{Seg:DE} and \eqref{Seq:JetE} respectively.
\end{Thm}

\begin{proof}
The  proof of Theorem \ref{k+1Leibniz} still holds until using the assumption $\mathrm{rank} E\geq 2$. So we also get an $(n+1)$-Lie bracket on $\Gamma(E)$ if the graph $\mathcal{G}_{\Pi^\sharp}$ is integrable. This bracket is local.
When $\mathrm{rank} E=1$, the map $\rho:\wedge^n \Gamma(E)\to \mathfrak{X}^1(M)$ is not necessarily a bundle map. So in general, we obtain a  local $(n+1)$-Lie algebra on $E$, which is an $(n+1)$-Lie algebroid if and only if \eqref{anchor linear} holds, namely, $\mathbbm{j} \circ \Pi^\sharp \circ \mathbbm{e}=0$. The converse is easy to get.
\end{proof}

Now we study in detail the case of $E=M\times \mathbbm{R}$, the trivial line bundle over $M$ and build the relation with the Nambu-Jacobi structures. For this, we first make some preparations.
In this case,
$\dev E=TM\times \mathbbm{R}$ and $\jet E=T^*M\times \mathbbm{R}$. Here $\dev E$ is a Lie algebroid with the Lie bracket and the anchor given by
\[[X+f,Y+g]=[X,Y]+Xg-Yf,\qquad \rho(X+f)=X,\qquad \forall~X,Y\in \mathfrak{X}^1(M),f,g\in C^\infty(M).\]
The pairing of $\dev E$ and $\jet  E$ is
\begin{eqnarray}\label{pairing for R}
(X+f,\xi+g)=\iota_X \xi+fg,\qquad \xi\in \Omega^1(M).
\end{eqnarray}
And the action of $\Gamma(\dev E)$ on $\Gamma(E)$ is $(X+g)f=Xf+gf$.

Before writing down the structures of the omni $n$-Lie algebroid $\dev E\oplus \wedge^n \jet E$ when $E=M\times \mathbbm{R}$ for a general $n$, we first make clear of the differential $\mathbbm{d}$ on $\jet_\bullet E$ in this case.

\begin{Lem}\label{ddd}
When $E=M\times \mathbbm{R}$, we have $\jet_n E=\Omega^n(M)\oplus \Omega^{n-1}(M)$ and the differential $\jetd: \jet_n E\to \jet_{n+1} E$ is given by
\begin{eqnarray*}
\mathbbm{d}(f)&=&df+f,\qquad \forall~f\in C^\infty(M)\\ \mathbbm{d}(\alpha_n+\alpha_{n-1})&=&d\alpha_n+\alpha_n-d\alpha_{n-1},\qquad \forall~\alpha_n\in \Omega^n(M), \alpha_{n-1}\in \Omega^{n-1}(M), n\geq 1.
\end{eqnarray*}
\end{Lem}
\begin{proof}
By the pairing of $\dev E$ and $\jet E$ and the action of $\dev E$ on $\Gamma(E)$ for $E=M\times \mathbbm{R}$, we have that $\jetd:\Gamma(E)\to \Gamma(\jet E)$ is
given by
\[\mathbbm{d}f(X+g)=(X+g)f=Xf+gf,\]
which implies that $\mathbbm{d}f=df+f$.  To see the action of  $\mathbbm{d}$ on $\Gamma(\jet E)$, for $\alpha\in \Omega^1(M)$, we have
\begin{eqnarray*}
\mathbbm{d}\alpha(X+f,Y+h)&=&(X+f)\iota_Y \alpha-(Y+h)\iota_X \alpha-\iota_{[X,Y]} \alpha\\ &=&d\alpha(X,Y)+f\iota_Y \alpha-h\iota_X \alpha.
\end{eqnarray*}
Hence we obtain $\mathbbm{d}\alpha=d\alpha+\alpha$. Here $\alpha\in \Omega^1(M)$ is viewed as an element in $\jet_2 E$ by $\alpha(f,Y)=f\iota_Y \alpha$ and $\alpha(Y,f)=-f\iota_Y \alpha$.

Then, for $g\in C^\infty(M)$ as a section of $\jet E$, we have
\begin{eqnarray*}
\mathbbm{d}g(X+f,Y+h)&=&(X+f)(gh)-(Y+h)(gf)-g(Xh-Yf)\\ &=& h(Xg)-f(Yg)\\ &=&dg(h,X)-dg(f,Y),
\end{eqnarray*}
thus we get $\mathbbm{d}g=-dg\in \Omega^1(M)$, which is seen as a section of $\jet_2 E$. The higher degrees are similar to get. 
\end{proof}

\begin{Lem}\label{lie der}
The Lie derivative of $\dev E=TM\times \mathbbm{R}$ on $\jet E=T^*M\times \mathbbm{R}$  and the contraction of $\jet E$ by $\dev E$ for $E=M\times \mathbbm{R}$ are given by
\begin{eqnarray}
\label{omni1}L_{X+f} (\xi+g) &=&L_X \xi+f\xi+gdf+Xg+fg;\\
\label{omni2}\iota_{X+f} \mathbbm{d} (\xi+g)&=&\iota_Xd\xi+f\xi-fdg-\iota_X \xi+Xg,
\end{eqnarray}
where $X\in \mathfrak{X}^1(M), \xi\in \Omega^1(M), f,g\in C^\infty(M)$.
\end{Lem}
\begin{proof}
By the general formula for the Lie derivative and the contraction of an omni-Lie algebroid, using \eqref{pairing for R}, we have
\begin{eqnarray*}
(L_{X+f} (\xi+g), Y+h)&=& (X+f) (\iota_Y \xi+gh)-(\xi+g,[X,Y]+Xh-Yf)\\ &=&X (\iota_Y \xi+gh)+f\iota_Y \xi+fgh-\xi([X,Y])-g(Xh-Yf)\\ &=&(L_X \xi+f\xi+gdf+Xg+fg,Y+h).
\end{eqnarray*}
So we get that
\[L_{X+f} (\xi+g)=L_X \xi+f\xi+gdf+Xg+fg.\]
By the fact that $\mathbbm{d}(g)=dg+g$
for any $g\in C^\infty(M)=\Gamma(E)$,
we obtain
\begin{eqnarray*}
\iota_{X+f} \mathbbm{d}(\xi+g)&=&L_{X+f} (\xi+g)-\mathbbm{d}(X+f,\xi+g)\\ &=&
L_X \xi+f\xi+gdf+Xg+fg-d(\iota_X \xi+fg)-\iota_X \xi-fg\\ &=&\iota_X d\xi+f\xi-fdg-\iota_X \xi+Xg.
\end{eqnarray*}
This finishes the proof.
\end{proof}

As a consequence, by using the Leibniz rule of the Lie derivative, the Lie derivative of $\dev E$ on $\wedge^n \jet E$ is also clear.

\begin{Pro}
For $E=M\times \mathbbm{R}$, the omni $n$-Lie algebroid $(\dev E\oplus \wedge^n \jet E,(\cdot,\cdot),\{\cdot,\cdot\},\rho)$ is as follows:
\begin{itemize}
\item{} $\dev E\oplus \wedge^n \jet E=TM\times \mathbbm{R}\oplus \big(\wedge^n T^*M\oplus \wedge^{n-1} T^*M\big)$,
\item{} the $(\wedge^{n-1} T^*M \oplus \wedge^{n-2} T^*M)$-valued pairing is given by
\begin{eqnarray}\label{special pairing}
( X+f, \alpha_n+\alpha_{n-1})=\iota_X\alpha_n+f\alpha_{n-1}-\iota_X \alpha_{n-1},\
\end{eqnarray}
\item{ } the bracket is given by
\begin{eqnarray}
\label{br1}\{X+f,Y+g\}&=&[X,Y]+Xg-Yf;\\
\label{br2}\{X+f,\alpha_n+\alpha_{n-1}\}&=&L_{X+f} (\alpha_n+\alpha_{n-1})\nonumber\\ &=&L_X\alpha_n+f\alpha_n+df\wedge \alpha_{n-1}+L_X \alpha_{n-1}+f\alpha_{n-1},\\
\label{br3}\{\alpha_n+\alpha_{n-1}, X+f\}&=&-\iota_{X+f}\jetd (\alpha_n+\alpha_{n-1})\nonumber \\ &=&-\iota_X d\alpha_n-f\alpha_n+fd\alpha_{n-1}+\iota_X \alpha_n-\iota_X d \alpha_{n-1},\end{eqnarray}
\end{itemize}
where $X,Y\in \mathfrak{X}^1(M),f,g\in C^\infty(M),\alpha_n\in \Omega^n(M),\alpha_{n-1}\in \Omega^{n-1}(M)$.\end{Pro}
\begin{proof}
By \eqref{pairing for R}, we have
 \[(X+f,\alpha_2)=\iota_X \alpha_2,\qquad (X+f,\alpha_1)=(X+f,1\wedge \alpha_1)=-\iota_X \alpha_1+f\alpha_1,\]
for $\alpha_2\in \Omega^2(M)$ and $\alpha_1\in \Omega^1(M)$, treating as sections of $\wedge^2 \jet E$.
Based on this idea, we get \eqref{special pairing}.
\eqref{br1} is clear.
By \eqref{omni1} and  the Leibniz rule, we get
\[\{X+f,\alpha_n+\alpha_{n-1}\}=L_{X+f} (\alpha_n+\alpha_{n-1})=L_X \alpha_n+L_X \alpha_{n-1}+f\alpha_n+df\wedge\alpha_{n-1}+f\alpha_{n-1}.\]
This is \eqref{br2}.
Similarly we have
\begin{eqnarray*}
\{\alpha_n+\alpha_{n-1},X+f\} &=&-L_{X+f} (\alpha_n+\alpha_{n-1})+\mathbbm{d}( \alpha_n+\alpha_{n-1},X+f)\\ &=&
-L_{X+f} (\alpha_n+\alpha_{n-1})+\mathbbm{d}(\iota_X \alpha_n+f\alpha_{n-1}-\iota_X \alpha_{n-1})
\\ &=&
-L_X\alpha_n-f\alpha_n-df\wedge \alpha_{n-1}-L_X \alpha_{n-1}-f\alpha_{n-1}\\ &&+d(\iota_X \alpha_n+f\alpha_{n-1})+\iota_X \alpha_n+f\alpha_{n-1}+d\iota_X \alpha_{n-1}\\ &=&-\iota_X d\alpha_n-f\alpha_n+fd\alpha_{n-1}+\iota_X \alpha_n-\iota_X d \alpha_{n-1},
\end{eqnarray*}
where the third identity follows from Lemma \ref{ddd}. Hence we get \eqref{br3}.
\end{proof}


\begin{Rm}
When $n=1$, we obtain the structure of the omni-Lie algebroid for $E=M\times \mathbbm{R}$:
\[\dev E\oplus \jet E=(TM\times \mathbbm{R})\oplus (T^*M\times \mathbbm{R}),\]
where the Dorfman bracket is
\begin{eqnarray*}
\{X+f,Y+g\}&=&[X,Y]+Xg-Yf;\\
\{X+f,\xi+g\}&=&L_X \xi+f\xi+gdf+Xg+fg;\\
\{\xi+g,X+f\}&=&-\iota_Xd\xi-f\xi+fdg+\iota_X \xi-Xg.
\end{eqnarray*}
for all $X,Y\in \mathfrak{X}^1(M), f,g\in C^\infty(M)$ and $\xi\in \Omega^1(M)$.

The skew-symmetrization of this bracket  is
\begin{eqnarray*}
{[X+f,Y+g]}&=&[X,Y]+Xg-Yf,\\
{[X+f,\xi+g]}&=&L_X \xi-\frac{1}{2}d\iota_{X} \xi+f\xi+\frac{1}{2}(gdf-fdg)+Xg-\frac{1}{2}\iota_X \xi+\frac{1}{2} fg.
\end{eqnarray*}
This is exactly  the bracket given in \cite{Wade} by Wade in the study of conformal Dirac structures.

Also, the Lie derivative and contraction in Lemma \ref{lie der} coincide with that in \cite{IW}, where they defined it directly.
\end{Rm}

Let $E=M\times \mathbbm{R}$ and $\Pi\in \Gamma(\dev^{n+1} E)$. The bundle map \[\Pi: \wedge^{n+1}\jet E=\wedge^{n+1} T^*M\oplus \wedge^{n} T^*M\to E=M\times \mathbbm{R}\]
 has two components
\[\Pi=\Lambda+\Gamma\in \mathfrak{X}^{n+1}(M)\oplus \mathfrak{X}^n(M).\]

As a consequence of  Theorem \ref{rank1case},  we have
\begin{Pro}
With the above notations, the graph of $\Pi^\sharp=\Lambda^\sharp+\Gamma^\sharp$ defines an integrable subbundle of the omni $n$-Lie algebroid $TM\times \mathbbm{R}\oplus (\wedge^{n}T^*M\times \mathbbm{R})$ if and only if it defines a Nambu-Jacobi structure of order $n+1$ on   $M$ whose Lie bracket is
\begin{eqnarray*}
[f_1,\cdots,f_{n+1}]&=&\Lambda(df_1,\cdots, d f_{n+1})+\sum_{i=1}^{n+1} (-1)^{i-1} f_i \Gamma(df_1,\cdots,\hat{df_i},\cdots df_{n+1}),
\quad f_i\in C^\infty(M).
\end{eqnarray*}
\end{Pro}
\begin{proof}
By definition,
Nambu-Jacobi structures on $M$ of order $n+1$ are local $(n+1)$-Lie algebra structures on the trivial line bundle $M\times \mathbbm{R}$. So by Theorem \ref{rank1case} and Lemma \ref{ddd}, we obtain a Nambu-Jacobi structure on $M$ with the bracket
\begin{eqnarray*}
[f_1,\cdots,f_{n+1}]&=&(\Lambda+\Gamma)(\jetd f_1,\cdots,\jetd f_n)\\ &=&
(\Lambda+1\wedge\Gamma)(df_1+f_1,\cdots,d f_n+f_n)\\ &=&\Lambda(df_1,\cdots,df_{n+1})+\sum_{i=1}^{n+1} (-1)^{i-1} f_i \Gamma(df_1,\cdots,\hat{df_i},\cdots df_{n+1}),
\end{eqnarray*}
which finishes the proof.
\end{proof}

This Nambu-Jacobi structure also appeared in \cite{H,MM} with a different sign convention. 

 Department of Applied Mathematics, China Agricultural University, Beijing, 100083, China

 Email:hllang@cau.edu.cn

 Department of Mathematics, Jilin University, Changchun, 130012, China

 Email:shengyh@jlu.edu.cn


\begin{thebibliography}{10}



\bibitem{BS} Y.  Bi and Y.  Sheng, On higher analogues of Courant algebroids, {\it Sci. China Math.} 54 (2011), 437-447.

\bibitem{BVZ}  Y.  Bi, L. Vitagliano  and T. Zhang, Higher omni-Lie algebroids, {\it J. Lie Theory} 29 (2019), 881-899.

\bibitem{BouwknegtJ}
P. Bouwknegt and B. Jur${\rm \check{c}}$o, AKSZ construction of topological open p-brane action and Nambu brackets, \emph{Rev. Math. Phys.} 25 (2013), 1330004, 31 pages.

\bibitem{BC} H. Bursztyn and A. Cabrera, Multiplicative forms at the infinitesimal level, {\it Math. Ann.}  353 (2012), 663-705.

\bibitem{CLomni}
Z. Chen and Z.  Liu,  Omni-Lie algebroids,  \emph{J. Geom. Phys.}  60 (2010), no. 5, 799-808.

\bibitem{CLS} Z. Chen, Z.   Liu and Y.  Sheng, $E$-Courant algebroids, {\it Int. Math. Res. Notices}  22 (2010), 4334-4376.
\bibitem{CLS2} Z. Chen, Z.   Liu and Y.  Sheng, Dirac structures of omni-Lie algebroids, {\it Int. J. Math.} 22 (2011), no. 8, 1163-1185.
\bibitem{CM} M. Crainic and I. Moerdijk, Deformation of Lie brackets: cohomological aspects, {\it J. Eur. Math. Soc.} 10 (2018), no. 4, 1037-1059.

\bibitem{DT} Y. Daletskii and L. Takhtajan, Leibniz and Lie algebra structures for Nambu algebra, {\it Lett. Math. Phys.} 39 (1997), 127-141.




 \bibitem{review}
J. A. de Azc$\rm\acute{a}$rraga and J. M. Izquierdo, $n$-ary algebras: a review with applications,
\emph{ J. Phys. A: Math. Theor.} 43 (2010), 293001.

\bibitem{Filippov}  V. T. Filippov, $n$-Lie algebras,  {\it Sib. Mat. Zh.} 26 (1985), 126-140.

\bibitem{Grabowski}
J. Grabowski, Brackets,  \emph{Int. J. Geom. Methods Mod. Phys.} 10 (2013), 1360001, 45 pages.

\bibitem{GM} J. Grabowski and G. Marmo, On Filippov algebroids and multiplicative Nambu-Poisson structures, {\it Diff. Geom. Appl.} 12 (2000), 35-50.

\bibitem{GS}
M. Grutzmann and T. Strobl,   General Yang-Mills type gauge theories for $p$-form gauge fields: from physics-based ideas to a mathematical framework or from Bianchi identities to twisted Courant algebroids, \emph{Int. J. Geom. Methods Mod. Phys.} 12 (2015), no. 1, 1550009, 80 pp.

\bibitem{hagiwara}
Y. Hagiwara, Nambu-Dirac manifolds, \emph{J. Phys. A: Math. Gen.} 35 (2002), no. 5, 1263-1281.

\bibitem{H} Y. Hagiwara, Nambu-Jacobi structures and Jacobi algebroids, {\it J. Phys. A: Math. Gen.} 37 (2004),  no.  26, 6713-6725.

\bibitem{hull}
C. M. Hull, Generalised geometry for M-theory, \emph{J. High Energy Phys.} 07 (2007), 079.


\bibitem{ILM} R. Inabez and M, de Leon abd J. C. Marrero and E. Padron, Leibniz algebroid associated with a Nambu-Poisson
structure, {\it J. Phys. A: Math. Gen.} 32 (1999), 8129-8144.

\bibitem{IL} R. Ibanez, B. Lopez, J. C. Marrero and E. Padron, Matched pairs of Leibniz algebroids, Nambu-Jacobi structures and modular class, {\it C. R. Acad. Sci., Paris} 333 (2001), 861-866.

\bibitem{XuC} D. Iglesias Ponte, C. Laurent-Gangoux and P. Xu, Universal lifting theorem and quasi-Poisson groupoids, {\it J. Eur. Math. Soc.} 14 (2012), no. 3, 681-731.

\bibitem{IW} D. Iglesias Ponte and A. Wade, Contact manifolds and generalized complex structures, {\it J. Geom. Phys.} 53 (2005), 249-258.


\bibitem{Lodayalgebroid}
D. Khudaverdian, J. Grabowski and N. Poncin, The supergeometry of Loday algebroids, \emph{J. Geom. Mech.} 5 (2) (2013), 185-213.

\bibitem{kinyon-weinstein}
M. K. Kinyon and A. Weinstein,
\newblock Leibniz algebras, {C}ourant algebroids, and multiplications on reductive homogeneous spaces,
\newblock {\em Amer. J. Math.} 123 (2001), no. 3,  525-550.

\bibitem{Kos}
Y. Kosmann-Schwarzbach, Courant algebroids. A short history, \emph{SIGMA Symmetry Integrability Geom. Methods Appl.} 9 (2013), 014, 8 pages.


\bibitem{LSW} H.   Lang, Y.  Sheng and A. Wade, $\VB$-Courant algebroids, $E$-Courant algebroids and generalized geometry, {\it Canadian Math. Bull.} 61 (2018), no. 3,  588-607.


\bibitem{LS} J.  Liu, Y.  Sheng and C.  Wang, Omni $n$-Lie algebras and linearization of higher analogues of Courant algebroids, {\it Int. J. Geom. Methods Mod. Phys.} 14 (2017), no. 7, 1750113, 18 pages.
    \bibitem{lwx}%
Z. Liu, A. Weinstein and P. Xu, Manin triples for Lie
bialgebroids, \emph{J. Diff. Geom.} 45 (1997), no. 3, 547-574.




\bibitem{MVV} G. Marmo, G. Vilasi and A. M. Vinogradov, The local structure of $n$-ary Poisson and $n$-Jacobi manifolds, {\it J. Geom. Phys.} 25 (1998), 141-182.

\bibitem{MM} K. Mikami and T. Mizutani, Foliations associated with Nambu-Jacobi structures, {\it Tokyo J. Math.} 28 (2005), no. 1, 33-54.

\bibitem{Sheng} Y.  Sheng, On deformation of Lie algebroids, {\it Results Math.} 62 (2012), 103-120.

\bibitem{ShengLiuZhu}
Y. Sheng, Z. Liu and C. Zhu, Omni-Lie 2-algebras and their Dirac structures, {\it J. Geom. Phys.} 61 (2011), no. 2,  560-575.

\bibitem{UchinoOmni}
K.~Uchino.
\newblock Courant brackets on noncommutative algebras and omni-{L}ie algebras,
\newblock {\em Tokyo J. Math.} 30 (2007), no. 1, 239-255.

\bibitem{Luca18}
L. Vitagliano,   Dirac-Jacobi bundles, \emph{J. Symplectic Geom.} 16 (2018), no. 2, 485-561.

\bibitem{Luca16}
L. Vitagliano and A. Wade,   Generalized contact bundles, \emph{ C. R. Math. Acad. Sci. Paris} 354 (2016), no. 3, 313-317.

\bibitem{Wade} A. Wade, Conformal Dirac structures, {\it Lett. Math. Phys. } 53 (2000), 331-348.

\bibitem{Alan}
A. Weinstein, Omni-Lie algebras, Microlocal analysis of the Schrodinger equation and related topics (Japanese) (Kyoto, 1999), \emph{S${\bar{u}}$rikaisekikenky${\bar{u}}$sho K${\bar{u}}$ky${\bar{u}}$roku}, 1176 (2000), 95-102.

\bibitem{Zambon}
M. Zambon, $L_\infty$-algebras and higher analogues of Dirac structures and Courant
algebroids, \emph{J. Symplectic Geom.} 10 (2012),  no. 4, 563-599.
\end{thebibliography}
\end{document}